%% file: arxiv-draft.tex
\documentclass[10pt,twocolumn,twoside]{IEEEtran}
\IEEEoverridecommandlockouts
\ifCLASSINFOpdf
  \usepackage[pdftex]{graphicx}
  \graphicspath{{../figures/}}
\else
\fi

\usepackage{url}
\usepackage{amsmath}
\usepackage{latexsym,amssymb}
\usepackage[tight,footnotesize,center]{subfigure}
\usepackage{amssymb,amsfonts,amsmath,amstext,amsthm}
\usepackage{cleveref}
\usepackage{latexsym,epsfig,psfrag,graphicx, color,subfigure}

\usepackage{tabularx}
\usepackage{algorithm}
\usepackage{algorithmic}

\newcommand{\law}{\ensuremath{\mathcal{L}}}
\newcommand{\F}{\mathcal{F}}
\newcommand{\bb}[1]{\mathbb{#1}}
\newcommand{\pr}{\mathbb{P}}

\renewcommand{\b}{\beta}
\renewcommand{\r}{\rho}
\renewcommand{\t}{\theta}
\newcommand{\E}{\mathbb{E}}

\newcommand{\e}{\epsilon}
\newcommand{\ignore}[1]{}
\newcommand{\BEQA}{\begin{eqnarray}}
\newcommand{\EEQA}{\end{eqnarray}}
\newcommand{\nn}{\nonumber}
\newtheorem{theorem}{Theorem}
\newtheorem{lemma}[theorem]{Lemma}

\newtheorem{corollary}[theorem]{Corollary}
\newtheorem{definition}{Definition}
\newtheorem{assumption}{Assumption}

\begin{document}

\title{A Mean Field Game Approach to\\ Scheduling in Cellular Systems}
% A Mean Field Game Approach to\\ Scheduling in Cellular Systems}
%\thanks{Research was funded in part by NSF grants CNS-0904520, CNS-0963818, CNS-1149458, DTRA grant HDTRA 1-13-1-0030 and AFOSR FA9550-13-1-0008.  Any opinions, findings, and conclusions or recommendations expressed in this material are those of the authors and do not necessarily reflect the views of the funding agencies.} \thanks{978-1-4799-3360-0/14/\$31.00~\copyright~2014 IEEE}}
%\numberofauthors{3}
\author{\IEEEauthorblockN{Mayank~Manjrekar\IEEEauthorrefmark{1},
Vinod~Ramaswamy\IEEEauthorrefmark{1},
Vamseedhar~Reddyvari\IEEEauthorrefmark{1}, and
Srinivas~Shakkottai\IEEEauthorrefmark{1}}\\
\IEEEauthorblockA{{\small{\IEEEauthorrefmark{1}Dept. of ECE, Texas A\&M University, Email: \{mayankm, vinod83, vamseedhar.reddyvaru, sshakkot\}@tamu.edu}}}
}

% \author{
% %\alignauthor
% Mayank Manjrekar,
% % \\\email{mayankm@tamu.edu}
% %\alignauthor
% Vinod Ramaswamy
% %  \\\email{vinod83@tamu.edu}
% %\alignauthor
% \& Srinivas Shakkottai
% %  \email{sshakkot@tamu.edu}
% }
%\IEEEpubid{978-1-4799-3360-0/14/\$31.00~\copyright~2014 IEEE}

\maketitle
%%%%%%%%%%%%%%%%%%%%%%
\input{00-abstract}

%%%%%%%%%%%%%%%%%%%%%%
\section{Introduction}\label{section:Introduction}
\input{01-intro}

%%%%%%%%%%%%%%%%%%%%%%
\section{Mean Field Game Model}\label{section:game}
\input{02-model}

%%%%%%%%%%%%%%%%%%%%%%
\section{Properties of Optimal Bid Function}\label{sec:opt-bid}
\input{03-bid-function}

%%%%%%%%%%%%%%%%%%%%%%
\section{Existence of MFE}\label{sec:exmfe}
The main result showing the existence of MFE is as follows.
\begin{theorem}
There exists an MFE $(\r,\hat{\t}_{\r})$ such that
\[ \r(x) = \gamma(x)\triangleq \Pi_{\r}\left(\hat{\theta}_{\r}^{-1}[0,x]\right), \forall x \in \mathbb{R}^{+}. \]
\end{theorem}
We first introduce some useful notation.  Let $\Theta = \{ \theta : \bb{R} \mapsto \bb{R}, \sup_{q\in \bb{R}^+}\left|\frac{\t(q)}{w(q)}\right|<\infty \}.$ Note that $\Theta$ is a normed space with $w$-norm.
Also, let $\Omega$ be the space of absolutely continuous probability measures on $\bb{R}^+$. We endow this probability space with the topology of weak convergence. Note that this is same as the topology of point-wise convergence of continuous cumulative distribution functions.

We define $\t^*:\mathcal{P}\mapsto \Theta $ as $(\t^*(\rho))(q) = \hat{\t}_\r(q)$, where $\hat{\t}_\r(q)$ is the optimal bid given by Corollary~\ref{cor:optbid}. It can easily verified that $\hat{\t}_\r \in \Theta$. Also, define the mapping ${\Pi}^*$ that takes a bid distribution $\r$ to the invariant workload distribution $\Pi_{\r}(\cdot)$.  Later, using Lemma~\ref{lemma:omega} we will show that $\Pi_{\r}(\cdot) \in \Omega$. Therefore, ${\Pi}^* : \mathcal{P} \rightarrow \Omega$. Finally, define $\mathcal{F}$ as $(\mathcal{F}(\r))(x) = \gamma(x) = \Pi_{\r} ({\hat{\theta}_\r}^{-1}([0,x]))$. \Cref{lem:PincludesF} will show that $\mathcal{F}$ maps $\mathcal{P}$ into itself. %For convenience, we use the notation $\mathcal{F}(x|\r)$ to denote $(\mathcal{F}(\r))(x)$.
%We will show that $\mathcal{F}(\mathcal{P}) \subset \mathcal{P}$.

Now to prove the above theorem we need to show that $\mathcal{F}$ has a fixed point, i.e $\mathcal{F}(\r) =\r$. %Schauder's fixed point theorem, stated below, yields the sufficient conditions for the existence of a fixed point to the mapping $\mathcal{F}$.
\begin{theorem}[Schauder Fixed Point Theorem]\label{schauder}
Suppose $\mathcal{F}(\mathcal{P}) \subset \mathcal{P}$. If $\mathcal{F}(\cdot)$ is continuous, and $\mathcal{F}(\mathcal{P})$ is contained in a convex and compact subset of $\mathcal{P}$, then $\mathcal{F}(\cdot)$ has a fixed point.
\end{theorem}
In next section, we show that the mapping $\mathcal{F}$ satisfies the conditions of the above theorem, and hence it has a fixed point. Note that $\mathcal{P}$ is a convex set. Therefore, we just need to show that the other two conditions are satisfied.

\section{MFE Existence: Proof}  \label{sec:mfe-proof}
\subsection{Continuity of the map $\mathcal{F}$}
To prove the continuity of mapping $\mathcal{F}$, we first show that $\t^*$ and ${\Pi}^*$ are continuous mappings. To that end, we will show that for any sequence $\r_{n} \rightarrow \r$ in uniform norm,  we have $\t^*(\r_n) \rightarrow \t^*(\r)$ in $w$-norm and ${\Pi}^*(\r_n) \Rightarrow {\Pi}^*(\r)$ (where $\Rightarrow$ denotes weak convergence). Finally, we use the continuity of $\t^*$ and ${\Pi}^*$ to prove that $\mathcal{F}(\r_n) \rightarrow \mathcal{F}(\r).$% which completes the proof.

\vspace{0.1in}%
\noindent\emph{Step 1: Continuity of $\t^*$}%\label{sec:value}
\begin{theorem}\label{thm:conttheta}
The map ${\theta}^*$ is continuous.
\end{theorem}
\begin{proof}
%We outline the important steps in the proof. For details refer to \cref{sec:prfconttheta}.
Define the map $ V^*:\mathcal{P}\mapsto\mathcal{V}$ that takes $\r$ to $\hat{V}_\r(\cdot)$.   We begin by showing that $\|\hat\t_{\r_1}-\hat\t_{\r_2}\|_w\leq K\|\hat{V}_{\r_1}-\hat{V}_{\r_2}\|_w$, which means that the continuity of the map $V^*$ implies the continuity of the map $\t^*$.

Then we show two simple properties of the Bellman operator.  The first is that for any $\r\in\mathcal{P}$ and $f_1,\ f_2\in\mathcal{V}$,
\begin{align}\label{eq:A}
 \|T_\r f_1-T_\r f_2\|_w\leq \hat K\|f_1-f_2\|_w %\Rightarrow \quad \mbox{(A)}
\end{align}
for some large $\hat K$, independent of $\r$.  This result is available in Lemma~\ref{lem:theta-bound}  in Appendix~\ref{sec:step1}.

Second, let $T_{\r_1}$ and $T_{\r_2}$ be the Bellman operators corresponding to $\r_1,\ \r_2 \in \mathcal{P}$ and let $f \in\mathcal{V}$. We show that
\begin{align}
 \|T_{\r_1}f-T_{\r_2}f\|_w  \leq 2(M-1)K_1\|f\|_w\|\r_1-\r_2\|. \label{eq:B} %\Rightarrow~\mbox{(B)}
 \end{align}
This result is available in Lemma~\ref{lem:bellman-bound}  in Appendix~\ref{sec:step1}.

We then have
\begin{align}
&\|T_{\r_1}^j\hat{V}_{\r_2}-T_{\r_2}^j\hat{V}_{\r_2}\|_w\leq \|T_{\r_1}^j\hat{V}_{\r_2}-T_{\r_1}^{j-1}T_{\r_2}\hat{V}_{\r_2}\|_w \nonumber\\
& \quad\quad+ \|T_{\r_1}^{j-1}T_{\r_2}\hat{V}_{\r_2}-T_{\r_1}^{j-2}T_{\r_2}^2\hat{V}_{\r_2}\|_w+\cdots \nonumber\\
& \quad\quad +\|T_{\r_1}T_{\r_2}^{j-1}\hat{V}_{\r_2}-T_{\r_2}^j\hat{V}_{\r_2}\|_w\nonumber\\
& \leq \hat K^{j-1}\|T_{\r_1}\hat{V}_{\r_2}-T_{\r_2}\hat{V}_{\r_2}\|_w+\cdots \nonumber\\
&\quad\quad +\|T_{\r_1}T_{\r_2}^{j-1}\hat{V}_{\r_2}-T_{\r_2}^j\hat{V}_{\r_2}\|_w\label{eq:D}\\
&\leq (\hat K^{j-1}+\cdots+1)\|T_{\r_1}\hat{V}_{\r_2}-T_{\r_2}\hat{V}_{\r_2}\|_w\nn\\
&\leq 2(M-1)K\|\r_1-\r_2\|(\hat K^{j-1}+\cdots+1)\|\hat{V}_{\r_2}\|_w\label{eq:E}
 \end{align}
Here, (\ref{eq:D}) and (\ref{eq:E}) follow from (\ref{eq:A})  and (\ref{eq:B}), respectively.

Now, let $j$ be such that $T_{\r_1}^j$ is an $\alpha$-contraction, which is guaranteed to exist by \cref{lem:optimalityMDP}. Note that Statement~1 of Lemma \ref{lem:optimalityMDP} implies that such a $j<\infty$ exists.  Then we have
 \begin{align}
\|\hat{V}_{\r_1} &-\hat{V}_{\r_2}\|_w \\
 & =\|T_{\r_1}^j\hat{V}_{\r_1}-T_{\r_2}^j\hat{V}_{\r_2}\|_w\nonumber\\
 &\leq  \|T_{\r_1}^j\hat{V}_{\r_1}-T_{\r_1}^j\hat{V}_{\r_2}\|_w + \|T_{\r_1}^j\hat{V}_{\r_2}-T_{\r_2}^j\hat{V}_{\r_2}\|_w\nonumber\\
 &\implies (1-\alpha)\|\hat{V}_{\r_1}-\hat{V}_{\r_2}\|_w\leq \|T_{\r_1}^j\hat{V}_{\r_2}-T_{\r_2}^j\hat{V}_{\r_2}\|_w\label{eq:C}
 \end{align}

Finally, from (\ref{eq:E}) and (\ref{eq:C}), we get
\begin{align}
 \|\hat{V}_{\r_1}&-\hat{V}_{\r_2}\|_w \\
 & \leq \frac{2(m-1)K(\hat K^{j-1}+\cdots+1)\|\r_1-\r_2\|}{1-\alpha}\|\hat{V}_{\r_2}\|_w\nonumber\\
& \leq \ \frac{2(m-1)K(\hat K^{j-1}+\cdots+1)\|\r_1-\r_2\|}{1-\alpha}\nonumber\\
&\quad\quad \times (\|\hat{V}_{\r_1}\|_w+\|\hat{V}_{\r_1}-\hat{V}_{\r_2}\|_w).\nonumber
 \end{align}
Therefore, if $\frac{2(m-1)K(\hat K^{j-1}+\cdots+1)}{1-\alpha}\|\r_1-\r_2\|<\frac{1}{2}$, then
\begin{align}
&\|\hat{V}_{\r_1}-\hat{V}_{\r_2}\|_w \\
& \qquad \leq \frac{4(m-1)K(\hat K^{j-1}+\cdots+1)}{1-\alpha}\|\hat{V}_{\r_1}\|_w\|\r_1-\r_2\|\nonumber
\end{align}
Hence, the maps $ V^*$ and $\t^*$ are continuous.
\end{proof}

\vspace{0.1in}
\noindent\emph{Step~$2$: Continuity of the map ${\Pi}^*$}\\
Let $\Pi_{\r,\t}(.)$ be the invariant distribution generated by any $\t.$  Recall that ${\Pi}^*$ takes $\rho \in \mathcal{P}$ to probability measure $\Pi_{\r}(.) = \Pi_{\r, \hat{\theta}_{\r}}(.)$.    First, we show that  $\Pi_{\r,\t}(.) \in \Omega,$ where $\Omega$ is the space of absolutely continuous measures  (with respect to Lebesgue measure) on $\bb{R}^{+}$.
%\subsection{Properties of invariant Distribution}\label{sec:distribution}
%Suppose the agent follows a stationary strategy $\t\in\Theta$. We shall now prove that the invariant distribution $\Pi_{\r,\t}(\cdot)$ generated by the dynamics in \ref{eq:trprobs} is absolutely continuous with respect to Lebesgue measure.
\begin{lemma}\label{lemma:omega}
 For any $\r\in\mathcal{P}$ and any $\t\in\Theta$, $\Pi_{\r,\t}(\cdot)$ is absolutely continuous with respect to the Lebesgue measure on $\bb{R}^+$.
\end{lemma}
 \begin{proof}
%$\Pi_{\r,\t}(\cdot)$
$\Pi_{\r,\t}(\cdot)$ can be expressed as the invariant queue-length distribution of the dynamics
\begin{align}\nonumber
   q\rightarrow  \begin{cases}
                        Q'+A & \textrm{with probability } \b \\
                        R &\textrm{with probability }(1- \b),
                       \end{cases}
  \end{align}
where $A\sim\Phi$ and $R\sim\Psi,$ and $Q'$ is a random variable with distribution generated by the conditional probabilities
\begin{align}
 \mathbb{P}(Q'=q|q)=& 1-p_\r(\hat{\t}(q))\nonumber\\
 \mathbb{P}(Q'=(q-1)^+|q)=& p_\r(\hat{\t}(q))  \nonumber
\end{align}
Let $\Pi'$ be the distribution of $Q'$. Then for any Borel set $B$, $\Pi$ can be expressed using the convolution of $\Pi'$ and $\Phi:$
\begin{align}
\Pi_{\r,\t}(B) =\b\int_{-\infty}^\infty \Phi(B-y)d\Pi'(y)+(1-\b)\Psi(B).
\end{align}
If $B$ is a Lebesgue null-set, then so is $B-y\ \ \forall y$. So, $\Phi(B-y)=0$ and $\Psi(B)=0$ and therefore $\Pi_{\r}(B)=0$.
 \end{proof}

We now develop a useful characterization of $\Pi_{\r,\t}$. Let \[\Upsilon_{\r,\t}^{(k)}(B|q)=\pr(Q_k\in B|\textrm{no regeneration,}~Q_0=q)\] be the distribution of queue length $Q_k$ at time $k$ induced by the transition probabilities (\ref{eq:trprobs}) conditioned on the event that $Q_0=q$ and that there are no regenerations until time $k$. We can now express the invariant distribution $\Pi_{\r,\t}(\cdot)$ in terms of $\Upsilon_{\r,\t}^{(k)}(\cdot|q)$ as in the following lemma.

\begin{lemma}\label{lem:queuelengthdist}
 For any bid distribution $\r\in\mathcal{P}$ and for any stationary policy $\t\in\Theta$, the Markov chain described by the transition probabilities in (\ref{eq:trprobs}) has a unique invariant distribution $\Pi_{\r,\t}(\cdot)$. Also $\Pi_{\r,\t}$ and $\Upsilon_{\r,\t}^{(k)}$ are related as follows:
 \begin{align}
  \Pi_{\r,\t}(B)=\sum_{k\geq 0}(1-\b)\b^k\E_{\Psi}(\Upsilon_{\r,\t}^{(k)}(B|Q)),
 \end{align}
 where $\E_{\Psi}(\Upsilon_{\r,\t}^{(k)}(B|Q)) = \int \Upsilon_{\r,\t}^{(k)}(B|q) d\Psi(q)$.
\end{lemma}
\begin{proof}
$\Upsilon_{\r,\t}^{(k)}(B|q)$ is the queue length distribution assuming no regeneration has happened yet, and the regeneration event occurs with probability $\beta$ independently of the rest of the system.  It is then easy to find  $\Pi_{\r,\t}(B)$ in terms of $\Upsilon_{\r,\t}^{(k)}(B|q)$ by simply using the properties of the conditional expectation, and the theorem follows.    Note that in $\E_{\Psi}(\Upsilon_{\r,\t}^{(k)}(B|Q)),$ the random variable is the initial condition of the queue, as generated by $\Psi.$  Full details are available in Appendix~\ref{sec:step2}.
\end{proof}

%\subsection{Continuity of $\Pi_\rho$ with respect to $\rho$}
%Now consider the space of all absolutely continuous probability measures $\Omega$ on $\bb{R}^+$ and endow it with the topology of weak convergence; which is the same as topology of point-wise continuity of the cumulative distribution functions. Consider the mapping $\hat\Pi:\mathcal{P}\mapsto\Omega$ that takes a bid distribution $\r$ to the invariant workload distribution $\Pi_{\r}(\cdot)=\Pi_{\r,\t_\r}(\cdot)$ when agent plays  using the optimal policy $\t_\r$. Also, let $\Upsilon_{\r}^{(k)}=\Upsilon_{\r,\t_\r^*}^{(k)}$.

We shall now prove the continuity of $\Pi^*$ in $\r$.  %Let $\Upsilon_{\r}^{(k)}=\Upsilon_{\r,\hat{\t}_\r}^{(k)}$.
\begin{theorem}\label{thm:contPi}
 The mapping $\Pi^*:\mathcal{P}\mapsto\Omega$ is continuous.
\end{theorem}
 \begin{proof}
By Portmanteau theorem \cite{billingsley2009convergence},  we only need to show that for any sequence $\r_n\rightarrow \r$ in $w$-norm and any open set $B$, $\liminf_{n\rightarrow \infty}\Pi_{\r_n}(B)\geq\Pi_\r(B)$.
 By Fatou's lemma,
 \begin{align}
 \liminf_{n\rightarrow \infty}\Pi_{\r_n}(B) &=\liminf_{n\rightarrow\infty}\sum_{k=0}^\infty(1-\b)\b^k\E_{\Psi_R}[\Upsilon_{\r_n}^{(k)}(B|Q)]\nonumber\\
  &\geq\sum_{k=0}^\infty(1-\b)\b^k\E_{\Psi_R}[\liminf_{n\rightarrow\infty}\Upsilon_{\r_n}^{(k)}(B|Q)]
 \end{align}
where $Q\sim \Psi_R$.   Let $\Upsilon_{\r}^{(k)}=\Upsilon_{\r,\hat{\t}_\r}^{(k)}.$  We prove in \Cref{lem:liminfupsilon} (see Appendix~\ref{sec:step2}) that $\liminf_{n\rightarrow\infty}\Upsilon_{\r_n}^{(k)}(B|q)\geq\Upsilon_{\r}^{(k)}(B|q)$ for every $q \in \bb{R}^{+},$ and the proof follows.
 \end{proof}

\vspace{0.1in}
\noindent\emph{Step 3: Continuity of the mapping $\mathcal{F}$}\\
Now, using the results from Step~$1$ and Step~$2$, we establish continuity of the mapping $\mathcal{F}$. First, we show that $\mathcal{F}(\r) \in \mathcal{P}$.
\begin{lemma}\label{lem:PincludesF}
For any $\r\in\mathcal{P}$, let $\gamma(x)  = (\mathcal{F}(\r))(x) = \Pi_{\r}(\hat{\t}_\r^{-1}([0,x])), x \in \bb{R}^{+}$. Then, $\gamma \in \mathcal{P}$.
\end{lemma}
 \begin{proof}
From the definition of $\Pi_{\r}$, it is easy to see that $\gamma$ is a distribution function. Since $\hat{\t}_\r$ is continuous and strictly increasing function as shown in Lemma~\ref{lemma:rcontinuity}, $\hat{\t}_\r^{-1}(\{x\})$ is either empty or a singleton. Then, from Lemma~\ref{lemma:omega}, we get that $\Pi_{\r}(\hat{\t}_\r^{-1}(\{x\})) =0$. Together, we get that $\gamma(x)$ has no jumps at any $x$ and hence it is  continuous.

To complete the proof, we need to show that the expected bid under $\gamma(.)$ is finite.  In order to do this, we construct a new random process $\tilde{Q}_{k}$ that is identical to the original queue length dynamics ${Q}_{k}$, except that it never receives any service.   We show that this process stochastically dominates the original, and use this property to bound the mean of the original process by a finite quantity independent of $\r.$  Full details are presented in Appendix~\ref{sec:step3}.
\end{proof}

We now have the main theorem.% showing continuity of the map $\mathcal{F}$.
\begin{theorem}
 The mapping $\mathcal{F}:\mathcal{P}\mapsto\mathcal{P}$ given by $(\mathcal{F}(\r))(x)=\Pi_{\r}(\hat{\t}_\r^{-1}([0,x]))$ is continuous.
\end{theorem}
 \begin{proof}
Let $\r_n \rightarrow \r$ in uniform norm. From previous steps, we have $\hat{\t}_{\r_n} \rightarrow \hat{\t}_{\r}$ in $w$-norm and $\Pi_{\r_n} \Rightarrow \Pi_{\r}$. Then, using Theorem 5.5 of Billingsley \cite{billingsley2009convergence}, one can show that the push-forwards also converge: $$\Pi_{\r_n}(\hat{\t}_{\r_n}^{-1}(\cdot)) \Rightarrow \Pi_{\r}(\hat{\t}_{\r}^{-1}(\cdot)).$$
 Then, $\mathcal{F}(\r_n)$  converges point-wise to $\mathcal{F}(\r)$ as it is continuous at every $x$, i.e., $(\mathcal{F}(\r_n))(x) \rightarrow (\mathcal{F}(\r))(x)$ for all $x \in \bb{R}^{+}$.

It is easy to  show that in the norm space $\mathcal{P}$, point-wise convergence implies convergence in uniform norm. This result is proved in Lemma~\ref {lem:point-uniform}  in Appendix~\ref{sec:step3}.  This completes the proof.
 \end{proof}

\subsection{$\mathcal{F}(\mathcal{P})$ contained in a compact subset of $\mathcal{P}$}
We show that the closure of the image of the mapping $\mathcal{F}$, denoted by $\overline{\mathcal{F}(\mathcal{P})}$, is compact in $\mathcal{P}$. As $\mathcal{P}$ is a normed space, sequential compactness of any subset of $\mathcal{P}$ implies that the subset is compact. Hence, we just need to show that $\overline{\mathcal{F}(\mathcal{P})}$ is sequentially compact. Sequential compactness of a set $\overline{\mathcal{F}(\mathcal{P})}$  means the following: if $\{\r_n\}\in \overline{\mathcal{F}(\mathcal{P})}$ is a sequence, then there exists a subsequence $\{\r_{n_j}\}$ and $\r \in \overline{\mathcal{F}(\mathcal{P})}$ such that $\r_{n_j}\rightarrow \rho$. %We prove this in two parts; the first uses Arzela-Ascoli theorem to show that there is a continuous function $f$ such that $f_{n_j}\rightarrow f$ uniformly and the second uses uniform tightness of the measures in $F(\mathcal{P})$ to show the $f \in \mathcal{P}$.
We use Arzel\`{a}-Ascoli theorem and uniform tightness of the measures in $\mathcal{F}(\mathcal{P})$ to show the sequential compactness. The version that we will use is stated below:
\begin{theorem}[Arzel\`{a}-Ascoli Theorem]
  Let $X$ be a $\sigma$-compact metric space. Let $\mathcal{G}$ be a family of continuous real valued functions on $X$. Then the following two statements are equivalent:
  \begin{enumerate}
   \item For every sequence $\{g_n\}\subset \mathcal{G}$ there exists a subsequence $g_{n_j}$ which converges uniformly on every compact subset of $X$.
   \item The family $\mathcal{G}$ is equicontinuous on every compact subset of $X$ and for any $x\in X$, there is a constant $C_x$ such that $|g(x)|<C_x$ for all $g\in \mathcal{G}$.
  \end{enumerate}
\end{theorem}

Suppose a family of functions $\mathcal{D}\subseteq\mathcal{P}$ satisfies the equivalent conditions of the Arzel\'{a}-Ascoli theorem and in addition satisfy the uniform tightness property, i.e., $\forall \e>0$ there exists and $x_\e$ such that for all $f\in\mathcal{D}$ $1\geq f(x_\epsilon)\geq 1-\e$. Then, for any sequence $\{\r_n\}\subset\mathcal{D}$, there exists a subsequnce $\{\r_{n_j}\}$ that converges uniformly on every compact set to a continuous increasing function $\r$ on $\bb{R}^+$. As $\mathcal{D}$ is uniformly tight it can be shown that $\r_{n_j}$ converges uniformly to $\r$ and that $\r\in\mathcal{P}$. Therefore, $\overline{\mathcal{D}}$ is sequentially compact in the topology of uniform norm.

\ignore{Say the family of functions $\mathcal{F}(\mathcal{P})$ satisfy the equivalent conditions of Arzel\`{a}-Ascoli theorem. Also, let they satisfy the uniform tightness property, i.e, $ \forall f \in \mathcal{F}(\mathcal{P})$, there exists an $x_{\epsilon}$ such that $1\geq f(x_{\epsilon}) >1-\epsilon$. Then, for any sequence $\{\r_n\}\subset \mathcal{F}(\mathcal{P})$, there exists a subsequence $\{\r_{n_j}\}$ that converges uniformly on every compact sets to a continuous increasing function $\r$. As these functions are uniformly tight, uniform convergence on compact sets imply uniform convergence. i.e. $\r_{n_j} \rightarrow \r$. Therefore, $\overline{\mathcal{F}(\mathcal{P})}$ is  sequentially compact in the topology of uniform norm. %totally bounded and hence so is its closure.

Finally, we need to show that $\overline{\mathcal{F}(\mathcal{P})} \subset \mathcal{P}$. From the tightness property, the limit function $\r$ satisfies that $1\geq \r(x_{\epsilon})\geq (1-\epsilon)$ and therefore $r(\infty)=1$. Also, we have
\begin{align}
\int (1 - \r(x)) dx \leq \liminf_{n_{j} \rightarrow \infty} \int (1 -\r_{n_{j}}(x)) dx
< \infty.
\end{align}
The first inequality is due to Fatou's lemma. And, the second inequality holds since $\{\r_{n_j}\} \in \mathcal{P}$. Therefore $\r \in \mathcal{F}$ and hence $\overline{\mathcal{F}(\mathcal{P})} \subset \mathcal{P}$.
}

In the following, we show that $\mathcal{F}(\mathcal{P})$ satisfies uniform tightness property and condition $2$ in Arzel\'{a}-Ascoli theorem. First verifying the conditions of  Arzel\'{a}-Ascoli theorem, note that the functions in consideration are uniformly bounded by $1$. To prove equicontinuity, consider a $\gamma = \mathcal{F}(\rho)$ and let $x>y$.
\begin{align}\label{eq:temp1}
 \gamma(x)-\gamma(y)&=\Pi_\r(\t_\r(q)\leq x)-\Pi_\r(\t_\r(q)\leq y)\nonumber\\
&=\Pi_\r(y<\t_\r(q)\leq x)
\end{align}

\begin{lemma}\label{lem:bound-on-Pi}
For any interval $[a,b]$, $\Pi_\r([a,b])<c\cdot(b-a)$, for some large enough $c$.
\end{lemma}
\begin{proof}
The proof follows easily from our characterization of $\Pi_\r$ in terms of $\Upsilon^{(k)}_\r.$
\end{proof}

The above lemma and equation (\ref{eq:temp1}) imply that $\gamma(x)-\gamma(y)\leq c(\t_\r^{-1}(x)-\t_\r^{-1}(y))$. To show equicontinuity, it is enough to show that $\limsup_{y\uparrow x} \frac{\gamma(x)-\gamma(y)}{x-y}\leq K(x)$ for some $K$ independent of $\r$, which we will show now.
\begin{align*}
  \limsup_{y \uparrow x} \frac{\gamma(x)-\gamma(y)}{x-y} &= \limsup_{y \uparrow x} \frac{\Pi_{\rho} (y<\theta_{\rho}(q)\leq x) }{x-y} \\
  &= \limsup_{y \uparrow x} \frac{\Pi_{\rho} \left(\big[\theta_{\rho}^{-1}(y), \theta_{\rho}^{-1}(x)\big]\right) }{x-y} \\
  &\leq c \limsup_{y \uparrow x} \frac{\theta_{\rho}^{-1}(x) - \theta_{\rho}^{-1}(y)}{x-y}\\
  &= c \limsup_{y \uparrow x} \frac{\theta_{\rho}^{-1}(x) - \theta_{\rho}^{-1}(y)}{\theta_{\rho}\theta_{\rho}^{-1}(x) - \theta_{\rho}\theta_{\rho}^{-1}(y)}\\
\end{align*}
Let $x' = \theta_{\rho}^{-1}(x)$ and $y' = \theta_{\rho}^{-1}(y)$. Now,
\begin{align*}
  \limsup_{y \uparrow x} \frac{\gamma(x)-\gamma(y)}{x-y} &\leq  c \limsup_{y' \rightarrow  x'} \frac{x' - y'}{\theta_{\rho}(x') - \theta_{\rho}(y')}\\
  &= c \limsup_{y' \rightarrow  x'} \frac{x' - y'}{\beta \Delta V(x') - \beta \Delta V(y')}\\
  &\leq  c \limsup_{y' \rightarrow  x'} \frac{x' - y'}{\beta \left(\Delta C(x') - \Delta C(y')\right)}\\
  &\leq c \frac{1}{H(x')}
\end{align*}
Where,
\begin{align*}
 0<H(x')=\begin{cases}
       \E_A[C'(x'+A)-C'(\overline{x-1}+A)]& x'>1\\
       \E_A[ C'(x'+A)] & x'\leq 1
      \end{cases}
\end{align*}
and $C'(x) = \frac{dC(x)}{dx}$.
% \begin{align}
%  \limsup_{y\uparrow x} \frac{\hat{\r}(x)-\hat{\r}(y)}{x-y}\leq &c\limsup_{y\uparrow x}\frac{\t_\r^{-1}(x)-\t_\r^{-1}(y)}{x-y}\\
% =&c\limsup_{y\uparrow x}\frac{\t_\r^{-1}(x)-\t_\r^{-1}(y)}{\t_\r\t_\r^{-1}(x)-\t_\r\t_\r^{-1}(y)}\\
% \leq & c\limsup_{y'\rightarrow  x'}\frac{x'-y'}{\t_\r(x')-\t_\r(y')}\label{eq:temp2}\\
% \leq & c\limsup_{y'\rightarrow x'}\frac{x'-y'}{\b (\Delta C(x')- \Delta C(y'))}\\
% \leq & c\frac{1}{H(x')}
% \end{align}
% where $x'=\t_\r^{-1}(x))$. Here, \eqref{eq:temp2} is due to strict monotonicity of $\t_\r$ and
% \begin{align*}
%  0<H(x')=\begin{cases}
%        \E_A[C'(x'+A)-C'(\overline{x-1}+A)]& x'>1\\
%        \E_A[ C'(x'+A)] & x'\leq 1
%       \end{cases}
% \end{align*}
% and $C'(x) = \frac{dC(x)}{dx}$.

Finally, we have the following lemma showing that $\F(\mathcal{P})$ is uniformly tight.
\begin{lemma}\label{lem:tight}
$\F(\mathcal{P})$ is uniformly tight, i.e., for any $\epsilon >0$ and any $ f \in \F(\mathcal{P})$, there exists an $x_{\epsilon}\in \mathbb{R}$ such that $1-\epsilon \leq f(x_{\epsilon}) \leq 1$.
\end{lemma}
\begin{proof}
From Lemma~\ref{lem:PincludesF}, we have $\F(\mathcal{P}) \subseteq \mathcal{P}$. Hence, the expectation of the bid distributions in $\F(\mathcal{P})$ is bounded uniformly. An application of Markov inequality will give uniform tightness.
\end{proof}

\section{Approximation Results: PBE and MFE} \label{sec:approx}
In this section we prove that the mean field policy is an $\e$-Nash equilibrium.  We have the following theorem:
\begin{theorem}\label{thm:approxthm}
Let $(\r,\hat{\t}_\r)$ constitute an MFE. Suppose at time $0$ the queue length of the users is set independently across users according to $\Pi_\r$; and that their initial belief is also consistent.  Also, suppose that all queues except queue $1$ play the MFE policy $\hat{\t}_\r$. Then, for any policy $\t^N$ of queue $1$ that may be history dependent and any $q\in\mathbb{R}^+$, we have
\begin{align}
\limsup_{N\rightarrow\infty} V_{1,\mu_{1,0}}^{N}(q;\hat{\t}_\r,(\hat{\t}_\r)_{-1})-V_{1,\mu_{1,0}}^{N}(q;\t^N,(\hat{\t}_\r)_{-1})\leq 0,\nonumber
\end{align}
where $\mu_{1,0}=\Pi_\r$ and the superscript $N$ has been added to explicitly indicate the dependence on the number of cells.
%In other words, the queue does not lose too much, in the limit, by playing the MFE strategy.
\end{theorem}

The main idea behind the proof is a result called \emph{propagation of chaos,} and it identifies conditions under which any finite subset of the state variables are independent from each other.  We state this result now in our context.   We only provide brief sketches of proofs in this section, since the methodology is much the same as \cite{graham1994chaos} and space constraints do not allow us to present the full version of the proofs here.

\begin{lemma}[Propagation of chaos]\label{lem:approxlemma}
For any fixed indices $i_1,\ldots,i_k$, let $\law(Q_{i_1}^{N}(t),\ldots,Q_{i_k}^{N}(t)$ denote the probability law of the $k$-tuple of corresponding queues in the $MN$-queue system, at time $t$. Suppose that $\law((Q_{i_1}^{N}(0),\ldots,Q_{i_k,0}^{N}(0))=\otimes^k\Pi_\r$, where $(\r,\t_\r)$ is the solution to the MFE equation. Also, suppose that all queues are following mean field equilibrium strategy. Then for any $T>0$, we have
\begin{align*}
\law(Q_{i_1}^{N}(T),\ldots,Q_{i_k}^{N}(T)\Rightarrow \otimes^k\Pi_\r,
\end{align*}
as $N\rightarrow\infty$.
\end{lemma}
\begin{proof}
We shall only consider the case $k=2$; the proof of the general case is similar. We can follow the proof of Theorems 4.1 and 5.1 in Graham and Meleard \cite{graham1994chaos}. The proof is divided into two parts; the first part proves that for any two agents $i$ and $j$,
\begin{align*}
\|\law(Q_{i}^{N},Q_{j}^{N})-\law(Q_{i}^{N})\otimes\law(Q_{j}^{N})\|_D\rightarrow 0,
\end{align*}
where the subscript $D$ refers to the total variation norm.  In the second, we show that $\law(Q_{i}^{N})\Rightarrow\Pi_\r$.
Both parts rely on studying interaction graphs, defined in \cite{graham1994chaos}, which characterize the amount of interactions that any finite subset of agents may have had in the past.  %Details of this proof are presented in Appendix~\ref{sec:chaos}.
\end{proof}

The proof of \Cref{thm:approxthm} is as follows. Suppose we start at time $t=0$ with queue length of agent $1$ being $Q_1(0)$.  We can choose a time $T$ large enough so that the value added by auctions occurring after time $T$ is less than $\e$, due to discounting. Thus, the difference in value contributed by these auctions, when using policy $\t^{N}$ and $\hat{\t}_\r$ can be bounded by $2\e,$ and we can restrict attention  to the first $T$ auctions.

Using ideas similar to \Cref{lem:approxlemma}, we show that probability of the event $E^N$ that other agents that interact with agent $1$ at time $t$ have never been influenced by agent $1$  goes to $1$ as $N$ becomes large.   Thus, the belief of distribution of queue lengths of other agents encountered converges to $\Pi_\r$ according to \Cref{lem:approxlemma}. Then  we can show that for any $\e>0$ and $N$, %there exists $N$ such that for all $n>N$,
\begin{align*}
V^{N}_{1,\mu_{1,0}}(q;\t_\r,(\t_\r)_{-1})-V^{N}_{1,\mu_{1,0}}(q;\t^N,(\t_\r)_{-1})\leq \e,
\end{align*}
which yields the desired result.  %Full details are available in Appendix~\ref{sec:chaos}.

%%%%%%%%%%%%%%%%%%%%%%%%%%%%%%%%%%%%%%%%%%
\section{Simulation Results}\label{sec:sims}
\input{06-simulation}

%%%%%%%%%%%%%%%%%%%%%%%%%%%%%%%%%%%%%%%%%%
\section{Discussion and Concluding Remarks}\label{sec:concl}
\input{07-conclusion}

%%%%%%%%%%%%%%%%%%%%%%%%%%%%%%%%%%%%%%%%%%
\bibliographystyle{IEEEtran}
\bibliography{references}

%%%%%%%%%%%%%%%%%%%%%%%%%%%%%%%%%%%%%%%%%%
\input{08-appendix}

\end{document}

%% file: 00-abstract.tex
\begin{abstract}
% We study auction-theoretic scheduling in cellular networks using the idea of a mean field equilibrium (MFE).  Here, agents model their opponents through an assumed distribution over their action spaces, and play the best response action against this distribution. We say that the system is at MFE if this best response action turns out to be a sample drawn from the assumed distribution.  In our setting, the agents are smart phone apps that generate service requests, have costs associated with waiting, and bid against each other for service from base stations.  The users of the apps spend a geometrically distributed amount of time on each app, and then move on to another.  We show that in a system in which we conduct a second-price auction at each base station and schedule the winner at each time, there exists an MFE that will schedule the user with highest value at each time.  We further show that the scheme can be interpreted as a longest-queue-first type policy.    The result suggests that auctions can implicitly attain the same desirable results as queue-length based scheduling.  We also present results on the convergence between a system with a finite number of agents to a mean field case as the number of agents become large.  Finally, we show simulation results illustrating the simplicity of computation of the MFE in our setting.
%Current cellular networks neither provide for scheduling priorities based on the impact on QoE of different apps, nor allow end users to declare which apps' performance is most valuable to them.   
We study auction-theoretic scheduling in cellular networks as a means to enable such value declarations.  Our analysis is based on using the idea of mean field equilibrium (MFE).  Here, agents model their opponents through a distribution over their action spaces and play the best response.  The system is at an MFE if this action is itself a sample drawn from the assumed distribution.  In our setting, the agents are smart phone apps that generate service requests, experience waiting costs, and bid for service from base stations.  We show that if we conduct a second-price auction at each base station, there exists an MFE that would schedule the app with the (weighted) longest queue at each time.  The result suggests that auctions can attain the same desirable results as queue-length-based scheduling with full information.  We present results on the asymptotic convergence of a system with a finite number of agents to the mean field case, and conclude with simulation results illustrating the ease of computation of the MFE.

\end{abstract}

%% file: 01-intro.tex
There has recently been a rapid increase in the use of smart hand-held devices for Internet access.  These devices are supported by cellular data networks, which carry the packets generated by apps running on these smart devices.  These apps can be modeled as queues that arrive when the user starts the app, and depart when the user terminates that app.  Packets generated by an app are buffered in a queue corresponding to that app.  Queueing delays impact the quality of user experience (QoE) based on the app being used. % For example, the QoE of a streaming application might be quite sensitive to queueing delays.  
Users move around cells that each has a base station, and scheduling a particular user provides service to the queue that represents his/her currently running app.  User interest could shift from app to app, regardless of whether or not there are buffered packets.  Hence, a queue might terminate and be replaced by a new one even if there are jobs waiting for service.

%The problem that we consider is that of scheduling in such a system of many base stations and ephemeral queues.

%The problem of scheduling in wired and wireless systems has been a topic of much recent research.

An important problem in cellular data networks is that for scheduling, \emph{i.e.,} determining which queues receive service at each time instant.  Most work on scheduling has focused on the case of a finite number of infinitely long lived flows, with the objective being to maximize the total throughput.  A seminal piece of work under this paradigm introduced the so-called \emph{max-weight algorithm}  \cite{Tass}.  Here, the drift of a quadratic Lyapunov function is minimized by maximizing the queue-length weighted sum of  acceptable schedules.
 Later work (\emph{e.g.}, \cite{LinShr04,ErySri06,NeeMod05})  has used a similar approach in a variety of network scenarios.  If queues arrive and depart, then a natural scheduling policy in the single server case is a \emph{Longest-Queue-First (LQF)} scheme, in which each server serves the longest of the queues requesting service from it.  LQF has many attractive properties, such a minimizing the expected value of the longest queue in the system.  %It has been shown \cite{AbeMan12} that with Bernoulli arrivals this policy minimizes the longest queue and maximizes the shortest queue, giving rise to queues that are similar in length.

The above approach assumes that the queue length values are given to the scheduler, and that it is aware of the function that maps the queue-length to cost on QoE of queueing.  However, while the downlink queue lengths would naturally be available at a cellular base station, the only way to obtain uplink queue information is to ask the users themselves.  However, a larger value of queue length results in a higher probability of being scheduled under all the above policies, implying an incentive to lie.  

%Also, such a farmework does not allow for the human end user to decide his or her own priorities for the QoEs for different apps.  For instance, one user might place a high importance on the QoE of Skype and Facebook, whereas another might have high importance for YouTube and Pandora.  Thus, traditional approaches do not account for the human in the loop that is the ultimate consumer of the service.

An appealing idea is to use a pricing scheme to inform scheduling decisions for cellular data access.  These prices could be in the form of tokens issued by the cellular service provider that are used as currency in the service market.  An example of a pricing approach is presented in  \cite{HaSen12}, which describes an experimental trial of a system in which day-ahead prices are announced to users, who then decide on whether or not to use their 3G service based on the price at that time. However, these prices have to be determined through trial and error.  Can we determine prices by using an auction?

Our key objective is to design an incentive compatible scheduling scheme that behaves in a (weighted) LQF-like fashion.  We consider a system in which each app bids for service from the base station that it is currently connected to.  The auction is conducted in a second-price fashion, with the winner being the app that bids highest, and the charge being the value of the second highest bid.  It is well known that such an auction promotes truth-telling \cite{krishna2009auction}.  Would the scheduling decisions resulting from such auctions resemble that of LQF?  Would conducting such an auction repeatedly over time with queues arriving and departing result in some form of equilibrium?

\subsection*{Mean Field Games}

We investigate the existence of an equilibrium using the theory of Mean Field Games (MFG) \cite{LasLio07,TemBou09,adlakha2011equilibria,BorSun11,XuHaj12,IyeJoh14,LiBha15}.  In MFG, the players assume that each opponent would play an action drawn \emph{independently} from a fixed distribution over its action space.  The player then chooses an action that is the best response against actions drawn in this fashion.  We say that the system is at Mean Field Equilibrium (MFE) if this best response action is itself a sample drawn from the assumed distribution.   %Hence, the player's action should be consistent with the assumed distribution.

The MFG framework offers a structure to approximate so-called Perfect Bayesian Equilibrium (PBE) in dynamic games.  PBE requires each player to keep track of their beliefs on the future plays of every other opponent in the system, and play the best response to that belief.   This makes the computation of PBE intractable when the number of players is large.   %In our case, the state of each player is the current queue length, while the action is the bid made to the base station in which the phone is currently located.  A PBE would require each app to estimate the queue lengths and the associated bids of every other queue that it is competing against---something that is clearly quite hard.
%The MFG approximation would assume that the bids made by the opponents are drawn independently from some bid distribution, and to place a bid in response.
The MFG framework simplifies computation of the best response, and often turns out to be asymptotically optimal.

Work on MFGs has mostly focused on showing the existence, accuracy and stability of MFE \cite{LasLio07,TemBou09,adlakha2011equilibria,BorSun11}.   In the space of queueing systems, some recent work considers the game of sampling a number of queues and joining one \cite{XuHaj12}.  %where players do not know their payoffs exactly, and have to learn their true valuations \cite{iyer2011mean}.  %However, we are not aware of any previous work of MFE in the area of scheduling in communication systems.
However, ours is a scheduling problem in queueing system interacting with an auction system, which we believe is unique.

In the space of applications, Iyer \emph{et al.} \cite{IyeJoh14} study advertisers competing via a second price auction for spots on a webpage.  The bid must lie in a finite real interval, and the winner can place an ad on the webpage.  With time, the advertisers learn about the value of winning (probability of making a sale).   Li \emph{et al.} \cite{LiBha15} consider the problem of mechanism design for truthful state revelation (number of packets at each station) in a wireless D2D streaming system.  Their main result is to generalize the Groves mechanism to the mean field regime.  Both use some version of fixed point theorem to show existence of the MFE.

Neither of the above consider the use of auctions for service scheduling in queueing systems, which is the basis of our problem.   In our setup, the state is the queue with arrivals and departures, and we allow bids to lie in the full positive real line.  These considerations result in significant technical work to show existence and characterize the MFE.   In an earlier version of this work \cite{ManRam14}, we presented a preliminary version of our work without proofs.  The current work highlights the methodological contributions.

\subsection*{Overview of Paper}

We introduce the Mean Field Game (MFG) in Section~\ref{section:game}.  Here, a selected agent (app) has a belief $\rho$ about the bid distribution of the other agents, and assumes that their bids will be drawn independently from this distribution.  The state of the agent is its current queue length, and it faces a per-time step cost that is a function of its queue length, which models the impact on QoE of queueing delay.   The agent must place a bid based on the belief and its current state and belief about other agents.  

We consider the problem of determining the cost minimizing bid function and the corresponding value function as a Markov Decision Process in Section~\ref{sec:opt-bid}.   We show that the Bellman operator corresponding to the MDP is a contraction mapping with a unique minimum, implying that value iteration would converge to the best response bid.  Further, we show that the best response bid is monotone increasing in queue length.  We call the bid distribution across agents that results from playing the best responses as $\gamma.$

We next prove the existence of the MFE in Sections~\ref{sec:exmfe}--\ref{sec:mfe-proof}  by verifying the conditions of the Schauder Fixed Point Theorem.  We need to show that the mapping between the assumed bid distribution $\rho$ to the resultant bid distribution $\gamma$ is continuous, and that the space in which $\gamma$ lies is contained in a compact subset of the space from which $\rho$ is drawn.   In order to do this, we first show that the mapping between $\r$ and best response bid function is continuous, and then show that the map between $\r$ and the state (queue length) distribution is continuous.  Putting these together yields the required continuity conditions.  We then verify the conditions of the Arzel\`{a}-Ascoli Theorem for showing compactness.

We show in Section~\ref{sec:approx} that the MFE in our case is an asymptotically accurate approximation of a PBE.   The result follows from the fact that any finite subset of agents is unlikely to have interacted with any of the others as the number of agents becomes large.  Finally, we present simulation results in Section~\ref{sec:sims}, showing that MFE computation is straightforward.

\emph{Discussion:}  In the case of a single cost function (homogeneous QoE for all apps) the best response bid function is monotone increasing in the queue length regardless of $\rho.$  This implies that the service regime corresponding to MFE is identical to LQF (or a weighted version if we have different QoE classes for apps).  Further, our simulations suggest that if the base stations were to compute the empirical bid distribution and return it to the users, the eventual bid distribution would be the MFE.   Thus, the desirable properties of LQF are a natural result of  auction-based scheduling, while the queue length distribution would be that generated by LQF.

%% file: 02-model.tex
We consider a system consisting of $N$ cells and $NM$ agents (apps), which are randomly permuted in these cells at each time instant, with each cell having exactly $M$ agents\footnote{Note that our results are essentially unchanged if $M$ is an independent random variable with finite support.}.  The model is consistent on the likely evolution of the 5G cellular system, wherein we expect a large number of small, dense cells and much user mobility across different cells.  The mobility model is identical to the basic framework used in work on mobile wireless networks \cite{Gross02}.  Each cell contains a base station, which conducts a second price auction to choose which agent to serve.   Each agent must choose its bid in response to its state and its belief over the bids of its competitors.

\Cref{fig:diagram} illustrates the MFG approximation, which is accurate in the limit as $N$ becomes large.  An MFG is described from the perspective of a single candidate agent $i$, which assumes that the actions of all its competitors are drawn independently from some distribution.   The asymptotic accuracy of the independence assumption follows from a standard argument on the \emph{propagation of chaos} whose details are provided in Section~\ref{sec:approx}.  In \Cref{fig:diagram}, the auction (shown as blue/dark tiles) and the queue dynamics  (shown as beige/light tiles).
 Since we focus on a single agent, we do not need to specify its identity explicitly, unless we wish to compare its actions with those of other agents.  Hence, we will drop the index $i$ where possible for ease of notation.% In our case, we assume that there are exactly $M-1$ \emph{homogeneous} competitors.% and a second price auction is held at each discrete time instant.
\begin{figure}[ht]
\begin {center}
\vspace{-0.12in}
\includegraphics[width=3.4in]{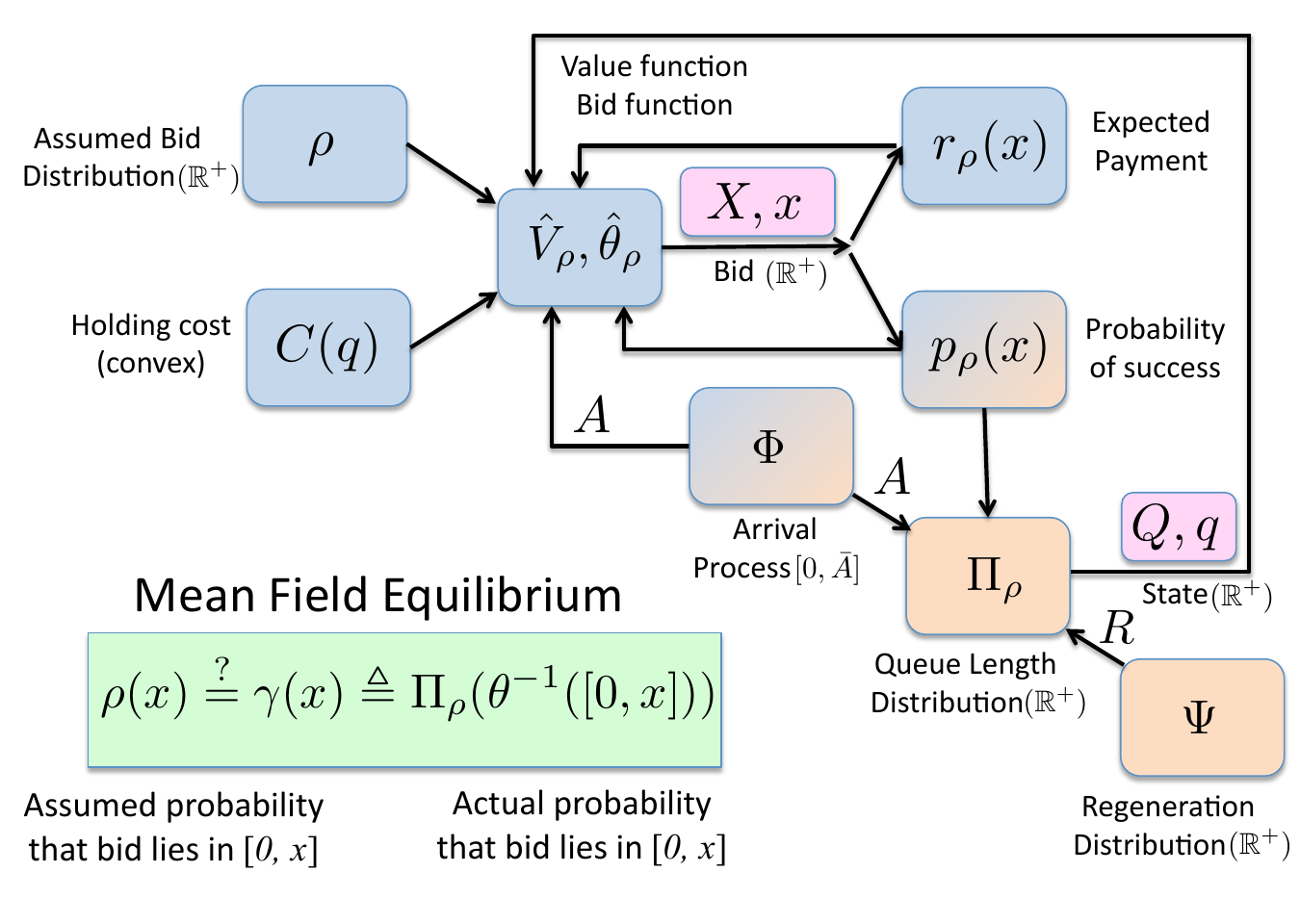}
\vspace{-0.15in}
\caption{The game consists of an auction part (blue/dark tiles) and a queue dynamics part (beige/light tiles). The system is at  MFE if the distribution of the bid $X$ is equal to the assumed bid distribution $\rho.$ }% The server creates blocks of content in realtime.}
\label{fig:diagram}
\vspace{-0.16in}
\end{center}
\end{figure}

\emph{Auction System:}  At each time step $k$, the agent of interest competes in a second price auction against $M-1$ other agents, whose bids are assumed to be independently drawn from a continuous, finite mean (cumulative) bid distribution $\rho,$ with support $\mathbb{R}^+.$  The state of the agent is its current queue length $q$ (the random variable is represented by $Q$).  The queue length induces a holding cost $C(q),$ where $C(.)$ is a strictly convex and increasing function.   The cost function could be one of a finite set of cost functions, modeling the impact on QoE of queueing for the currently running app for that agent.   However, since the analysis is identical for one or a finite set of cost functions, we focus on the single cost function in the analysis below.  We will discuss the (very minor) changes that result with multiple cost functions in Section~\ref{sec:concl}.

 Suppose that the agent bids an amount $x \in \mathbb{R}^+.$  The outcomes of the second price auction are that the agent would obtain one unit of service with probability $p_\r(x)$ and would have to pay an expected amount of $r_{\rho}(x)$ when all the other bids are drawn independently from $\rho.$  Further, the queue has future job arrivals according to distribution $\Phi,$ with the random job size being denoted by $A.$  Finally, the app can terminate at any time instant with probability $1-\beta.$   Based on these inputs, the agent needs to determine the value of its current state $\hat{V}_{\rho}(q),$ and the best response bid to make $x=\hat{\theta}_{\rho}(q)$.   The assumption that only a single unit of service is provided at each base station is for simplicity of notation, and our results are unchanged if we are allowed to choose some $\tilde{M} < M$ agents as winners in each auction.  The mechanism followed would be a $\tilde{M}+1^{th}$ price auction in that case.

%This problem can be formulated as a Markov Decision Process with discount factor $\beta$, with the value and bid function being determined by minimizing an appropriate Bellman Equation.  This process is described in \cref{}, where it is shown that The Bellman operation is a contraction mapping with a unique minimum, which implies that a value iteration would converge to the best response bid.  Further, we show that the best response bid is monotone increasing in queue length.

\emph{Queueing System:} The queueing dynamics are driven by the arrival process $\Phi$ and the probability of obtaining a unit of service being $p_\r(x)$ as described above.  When the user terminates an app, he/she immediately starts a fresh app, \emph{i.e.,}  a new queue takes the place of a departing queue.  The initial condition of this new queue $R$ is drawn from a regeneration distribution $\Psi,$ whose support is $\mathbb{R}^+.$  The invariant distribution associated with this queueing system (if it exists) is denoted $\Pi_\r.$

We make the following assumptions.
\begin{assumption}\label{assum:arrival}
At each time $k,$ the arrivals $\{A_{k}\}$ are i.i.d random variables distributed according to $\Phi$. We assume that $A_{k} \in [0, \{A\}],$ where $\bar{A}$ is finite.  Also, these random variables have a bounded density function, $\phi$ ({i.e.,} $||\phi\|<c_\phi,$ where $||.||$ is the sup norm).
\end{assumption}

\begin{assumption}\label{assum:regen}
The regeneration values $\{R_{k}\}$ are i.i.d random variables distributed according to $\Psi,$ and they have a bounded density $\psi$ ({i.e.,} $\|\psi\|<c_\psi,$ where $||.||$ is the sup norm).
\end{assumption}

\begin{assumption}\label{assum:cost}
The holding cost function $C: \bb{R}^{+} \mapsto \bb{R}^{+}$ is continuous, increasing and strictly convex. We also assume that $C$ is $O(q^m)$ for some integer $m$.
\end{assumption}
The polynomial form above is for technical reasons, but is not very restrictive since many convex functions can be approximated quite well.

\emph{Mean Field Equilibrium:}  The probability that the agent's bid lies in the interval $[0,x]$ is equal to the probability that the agent's queue length lies in some set whose best response is to bid between $[0,x].$  Thus, the probability of the bid lying in the interval $[0,x]$ is $\Pi_\rho(\hat{\theta}_{\rho}^{-1}([0,x])),$ which we define as $\gamma(x).$  According to the assumed (cumulative) bid distribution, the probability of the same event is $\rho(x).$  If $\r(x) = \gamma(x)$, it means that the assumed bid distribution is consistent with the best response bid distribution, and we have an MFE.

\subsection{Agent's decision problem}

Let the candidate be agent~$i$.  Suppose that the belief over the bid of a random agent has cumulative distribution $\rho$. We assume that $\r \in \mathcal{P}$ where,  $$\mathcal{P} = \{G| G~\mbox{is a continuous c.d.f}, \int (1- G(x)) dx < E  \},$$
 for some $E < \infty$, to be defined later.  Besides, its current state, the information available with the agent about the market at any time prior to the auction only includes the following:
\begin{enumerate}
 \item The bids it made in each previous auction from last regeneration.
 \item The auctions that it won.
 \item The payments made for the auctions won.
\end{enumerate}
Let $H_{i,k}$ be the history vector containing the above information available to agent~$i$ at time $k$.  Suppose that the random variable representing the bid made by agent $i$ at time $k$ is denoted by $X_{i,k},$ with the realized value being $x_{i,k}.$   Also, let $\bar{X}_{-i,k} = \max_{j \in M_{i,k}} X_{j,k},$ represent the maximum value of $M-1$  draws from the distribution $\rho.$  Thus, $\bar{X}_{-i,k}$ is the value of the highest opposing bid.  The agent's decision problem is to choose a \emph{bid function} $\theta_i,$ which maps its available information to a bid at each time $x_{i,k}.$

Since the time of regeneration $T_i^k$ is a geometric random variable, the expected cost of agent $i$ can be written as
\begin{align}\label{eq:mfvalue}
 V_{i,\r}(H_{i, k}; \theta_i)=\mathbb{E}\left[\sum_{t=k}^\infty \b^t [C(Q_{i,t})+r_\r(X_{i,t})]\right],
\end{align}
where the expectation is over future state evolutions.  By replacing the belief with $\rho$, we have made an agent's decision problem independent of other agents' strategies. Hence, we represent the cost by $V_{i,\r}(H_{i,k};\theta_i)$.    Also, $r_\r(x)=\mathbb{E}[\bar X_{-i,k}\mathbf{I}\{\bar X_{-i,k}\leq x\}]$ is the expected payment when $i$ bids $x$ under the assumption that the bids of other agents are distributed according to $\r$.  Hence, given $\r$, the win probability in the auction is
\begin{align}\label{eq:succprob}
 p_\r(x)=\pr(\bar{X}_{-i,k}\leq x)=\r(x)^{M-1}.
\end{align}
The expected payment when bidding $x$ is
% \BEQA\label{eq:exppayment}
%  r_\r(x)=\E[\bar{X}_{-i,k}\mathbf{I}\{\bar{X}_{-i,k}\leq x\}] = xp_\r(x)-\int_0^xp_\r(u)du.
% \EEQA
\BEQA\label{eq:exppayment}
 r_\r(x)=\E[\bar{X}_{-i,k}\mathbf{I}\{\bar{X}_{-i,k}\leq x\}]\nonumber\\
= xp_\r(x)-\int_0^xp_\r(u)du.
\EEQA

% Since, $T_i^k$ is a geometric random variable, the above expression reduces to
% \begin{align}\label{eq:mfvalue}
%  V_{i,\r}(H_{i, k}; \theta_i)=\mathbb{E}\left[\sum_{t=k}^\infty \b^t [C(Q_{i,t})+r_\r(X_{i,t})]\right].
% \end{align}

The state process $Q_{i,k}$ is Markov and has a transition kernel
\begin{align}\label{eq:trprobso}
& \pr(Q_{i, k+1}\in B|Q_{i,k}=q, X_{i,k} = x)= \nonumber\\
&\qquad \b p_\r(x)\pr((q-1)^++A_{k}\in B) \\
& \qquad+\b (1-p_\r(x))\pr(q+A_{k}\in B) + (1-\b)\Psi(B),\nonumber
\end{align}
where $B \subseteq R^{+}$ is a Borel set and $x^+\triangleq\max(x,0)$. Recall that $A_{k}\sim\Phi$ is the arrival between $(k)^{th}$ and $(k+1)^{th}$ auction and $\Psi$ is density function of the regeneration process. In the above expression, the first term corresponds to the event that agent wins the auction at time~$k,$ while the second corresponds to the event that it does not.  The last term captures the event that the agent regenerates after auction~$k$.   The agent's decision problem can be modeled as an infinite horizon discounted cost MDP. It can be shown that there exists an optimal Markov deterministic policy for our MDP  \cite{strauch1966negative}.  Then, from \eqref{eq:mfvalue}, the optimal value function of the agent is
\BEQA\label{eq:optvalue}
\hat{V}_{i,\r}(q)= \inf_{\theta_i \in \Theta} \mathbb{E}\left[\sum_{t=1}^\infty \b^t [C(Q_{i,t})+r_\r(X_{i,t})]\left|Q_{i,0} = q\right.\right],
\EEQA
where $\Theta$ is the space of Markov deterministic policies.

Note that user index is redundant in the above expression as we are concerned with a single agent's decision problem. In future notations, we will omit the user subscript~$i$.

\subsection{Invariant distribution}
Given cumulative bid distribution $\rho$ and a Markov policy $\theta \in \Theta$, the transition kernel given by \eqref{eq:trprobso} can be re-written as,
\begin{align}\label{eq:trprobs}
& \pr(Q_{k+1}\in B|Q_{k}=q)= \nonumber\\
&\qquad \b p_\r(\theta(q))\pr((q-1)^++A_{k}\in B) \\
& \qquad+\b (1-p_\r(\theta(q)))\pr(q+A_{k}\in B) + (1-\b)\Psi(B).\nonumber
\end{align}
%where $B \subseteq R^{+}$ is a Borel set and $x^+\triangleq\max(x,0)$. Recall that $A_{k}\sim\Phi$ is the arrival between $(k)^{th}$ and $(k+1)^{th}$ auction and $\Psi$ is density function of the regeneration process. In the above expression, the first two terms correspond to the event that the agent does not regenerate. In particular the first corresponds to the event that agent wins the auction at time~$k$. The last term captures the event that the agent regenerates after auction~$k$.
Then, we have an important result in the following lemma:
\begin{lemma}\label{lemma:stdist}
The Markov chain described by the transition probabilities in \eqref{eq:trprobs} is positive Harris recurrent and has a unique invariant distribution.
\end{lemma}
\begin{proof}
From \eqref{eq:trprobs} we note that, $\pr(Q_{k+1}\in B|Q_k=q)\geq (1-\b)\Psi(B),$ where $0<\b<1$ and $\Psi$ is a probability measure. The result then follows from results in Chapter 12, Meyn and Tweedie \cite{meyn2009markov}.
\end{proof}
We denote the unique invariant distribution by $\Pi_{\r,\t}$.

\subsection{Mean field equilibrium}
%In this section, we define the mean field equilibrium for our stochastic game.
%Assume that all agents conjecture the same bid distribution $\r$ and the decision problem in \ref{eq:optvalue} has an optimal policy $\hat{\t}_\r$. This induces a dynamics with transition probabilities as in \ref{eq:trprobs}. We have shown in the previous section that the dynamics induced by the transition kernel \ref{eq:trprobs} has a invariant distribution which we denote by $\Pi_{\r,\hat{\t}_\r}$.

As described in the Introduction, the mean field equilibrium requires the consistency check that the bid distribution $\gamma$ induced by the invariant distribution $\Pi_{\r,\t_\r}$ should be equal to the bid distribution conjectured by the agent, i.e., $\r$. %In other words we require,
% \begin{align}\label{eq:meanfield}
%  \r(x)=\Pi_{\r,\hat{\t}_\r}(\hat{\t}_\r^{-1}([0,x])).
% \end{align}
Thus, we have the following definition of MFE:
\begin{definition}[Mean field equilibrium]
Let $\rho$ be a bid distribution and $\theta_{\r}$ be a stationary policy for an agent. Then, we say that $(\rho, \theta_{\rho})$ constitutes a mean field equilibrium if
\begin{enumerate}
\item $\theta_{\rho}$ is an optimal policy of the decision problem in (\ref{eq:optvalue}), given bid distribution $\rho$; and
\item $ \r(x)=\gamma(x)\triangleq \Pi_{\r}(\t_{\r}^{-1}([0,x])), \forall x \in \bb{R^{+}}$, where $\Pi_\r=\Pi_{\r,\t_\r}$.
\end{enumerate}
\end{definition}
Note that the game theoretic definition of the MFE considers the existence of an invariant distribution at a fixed time as the number of agents becomes asymptotically large.  In keeping with extending the ideas of a Bayesian Nash Equilibrium to the system with a large number of agents, the definition does not require the occupancy distribution to converge to the invariant distribution from an arbitrary initial condition as time becomes large \cite{IyeJoh14,LiBha15,adlakha2011equilibria}.  The approach may be contrasted with the stochastic systems literature that often studies the mean field in the case where both the number of controllers and time go to infinity (in some order), but the control policy is fixed (not a best response), and the objective is to show that interchanging limits of time and number of particles produces the same steady state distribution \cite{BenBou08}.\\

A main result of this work is showing the existence of an MFE.  %Due to space limitations, we will only present proof sketches in most cases.  All proofs can be found in our technical report \cite{ManRam13-tech}.  %We begin by establishing the monotonicity  and continuity the optimal bid function. %These properties are essential in showing the existence of an MFE.

%% file: 03-bid-function.tex
%In this section, we state the optimality equation for the single agent's decision problem given in \eqref{eq:optvalue}, and determine the properties of the optimal strategy.
The decision problem given by \eqref{eq:optvalue} is an infinite horizon, discounted Markov decision problem. The optimality equation or Bellman equation corresponding to the decision problem is
\begin{align}\label{eq:bellman}
&\hat{V}_\r(q) = C(q) + \beta \E_A (\hat{V}_\r(q+A))+  \inf_{x \in R^{+}}\left[r_{\r}(x)\right.\nonumber\\
 & \left. - p_\r(x)\beta \E_A\left(\hat{V}_\r(q+A) - \hat{V}\r((q-1)^++A)\right)\right],
\end{align}
where $A\sim \Phi,$ and we use the notation $\max (0,z)=z^+.$ Note that the decision problem above is independent of the regeneration distribution $\Psi,$ since the game simply ends at any time with probability $1-\beta$ from the agent's perspective.  However, from a system perspective, the  Markov chain describing the state transition is correctly represented by \eqref{eq:trprobs}.

Define the set of functions
\begin{align}
\mathcal{V}= \left\{   f:\bb{R}^+\mapsto\bb{R}^+: \sup_{q\in \bb{R}^+}\left|\frac{f(q)}{w(q)}\right|<\infty \right\},
\end{align}
where $w(q) = \max\{C(q), 1\}$. Note that $\mathcal{V}$ is a Banach space with \emph{$w$-norm}, $$\|f\|_w=\sup_{q\in \bb{R}^+}\left|\frac{f(q)}{w(q)}\right|<\infty.$$  Also, define the operator $T_{\r}$ as
\begin{align}\label{eq:T}
&(T_{\r}f)(q)=C(q)+\beta \E_A f(q+A)+\inf_{x\in\bb{R}^+}\left[r_{\r}(x) \right.\nonumber\\
&\left. -p_{\r}(x)\b(\E_A (f(q+A)-f((q-1)^++A)))\right],
\end{align} where $f \in \mathcal{V}$.  It is straightforward to show that the infimum in the above operator occurs at
\BEQA
\beta \Delta f(q)^+,\label{eq:opt-bid}
\EEQA
 where $\Delta f(q) = \E_A (f(q+A)-f((q-1)^++A)).$  Then, substituting from (\ref{eq:succprob}), (\ref{eq:exppayment}) and (\ref{eq:opt-bid}), (\ref{eq:T}) can be rewritten as
%(Lemma~\ref{lemma:optimalbid})
\begin{align}\label{eq:Tnew}
 (T_{\r}f)(q) = C(q)+\beta \E_A f(q+A) - \int_{0}^{\beta \Delta f(q)^+} p_{\r}(u) du.
\end{align}

The following lemma characterizes the optimal solution.
\begin{lemma}\label{lem:optimalityMDP}
Given a cumulative bid distribution $\r$,
\begin{enumerate}
\item There exists a $j\in\bb{N}$ such that $T_\r^j:\mathcal{V}\rightarrow\mathcal{V}$ is a contraction mapping. Hence,  there exists a unique $f^*_{\r} \in \mathcal{V}$ such that $T_{\r}f^*_{\r} = f^*_{\r}$, and for any $f \in \mathcal{V}$, $T^n_{\r}f\rightarrow f^*_{\r}$ as $n\rightarrow\infty$.
\item The fixed point $f^*_{\r}$ of operator $T_{\r}$ is the unique solution to the optimality equation (\ref{eq:bellman}), i.e., $f^*_{\r} = \hat{V}_{\r}$.
\item Let, $\hat{\theta}_{\r}(q) = \beta \Delta \hat{V}_\r(q)^+.$ Then, $\hat{\theta}_{\r}$ is an optimal policy.
\end{enumerate}
\end{lemma}
\begin{proof}
The proof is similar to Theorem~$8.3.6$ in \cite{hernandez1999}. An exception to be noted here is that the action space in our case is not a compact set which violates Assumption~$8.3.1(a)$ in \cite{hernandez1999}. However, this assumption can be overridden if the statement of Lemma~$8.3.8 (a)$ in \cite{hernandez1999} holds true.  This applies to our case since, as derived in (\ref{eq:opt-bid}), a minimizer exists for the infimum operator in (\ref{eq:T}) for every $q.$    Further, Theorem~$8.3.6$ is specified with $j=1$ (a one-step contraction). Hence, we replace Assumption~$8.3.2(b)$ with an appropriate condition to obtain a $j-$step contraction. Please refer to Appendix~\ref{sec:optimalityMDP} for details.
\end{proof}

\begin{corollary}\label{cor:optbid}
An optimal policy of the agent's decision problem \eqref{eq:optvalue} is given by
$$\hat{\theta}_{\r}(q) = \beta \E_{A}\left[\hat{V}_\r(q+A) - \hat{V}_\r((q-1)^++A)\right].$$
\end{corollary}

We now establish that $\hat{V}_{\r}$ and $\hat{\t}_{\r}$ are continuous and increasing functions.
\begin{lemma}\label{lemma:rcontinuity}
Given a cumulative bid distribution function $\r$%, we have
\begin{enumerate}
 \item $\hat{V}_{\r}$ is a continuous increasing function.
 \item $\hat{\t}_{\r}$ is a continuous strictly increasing function.
\end{enumerate}
\end{lemma}
\begin{proof}
Let $f\in\mathcal{V}$. %Define the function $\Delta f$ as $\Delta f(q)=\E_A(f(q+A)-f((q-1)^++A))$.
Suppose $f$ is a continuous monotone increasing function. We first prove that $T_\r f$ is also  continuous monotone increasing function. Since, $T^n_\r f\rightarrow \hat{V}_{\r}$ according to Statement~2 of Lemma \ref{lem:optimalityMDP}, we conclude that $\hat{V}_\r$  also has the same property.

%First we prove that $T_{\rho}f$ is a monotone increasing function.
Let $q>q'$. Then,
\begin{align*}
&T_\r f(q)-T_\r f(q')= C(q)-C(q') \nonumber\\
&+ \b\E_A(f(q+A)-f(q'+A))\nonumber\\
& +\inf_x[r_\r(x)-\b p_\r(x)\E_A(f(q+A)-f((q-1)^++A))]\nonumber\\
&-\inf_x[r_\r(x)-\b p_\r(x)\E_A(f(q'+A)-f((q'-1)^++A))] \nonumber\\
% \end{align*}
% \begin{align*}
&\stackrel{(a)}{\geq}\b\E_A(f(q+A)-f(q'+A)) \nonumber\\
&\quad +\b\inf_x\left[p_\r(x)\E_A(f(q'+A)-f((q'-1)^++A) \right. \nonumber\\
&\qquad\qquad\left.-f(q+A)+f((q-1)^++A))\right] \nonumber\\
&\geq \b\min\left\{\E_A(f(q+A)-f(q'+A)), \right.\nonumber\\
&\quad\quad \left. \E_A(f((q-1)^++A)-f((q'-1)^++A))\right\} \stackrel{(b)}{\geq} 0,
\end{align*}
where (a) follows from the assumption that $C(.)$ is an increasing function, and (b) follows from the assumption that $f(.)$ is an increasing function.

To prove that $T_\r f$ is continuous consider a sequence $\{q_n\}$ such that $q_n\rightarrow q$. Since $f$ is a continuous function, $f(q_n+a)\rightarrow f(q+a)$. Then, by using dominated convergence theorem, we have $\E_Af(q_n+A)\rightarrow\E_Af(q+A)$ and $\E_Af((q_n-1)^++A)\rightarrow\E_Af((q-1)^++A)$. %It can also be shown that the optimal value of the minimization in $T_\r f(q_n)$ occurs at $\Delta f(q_n)=\E_A(f(q_n+A)-f((q_n-1)^++A))$ and therefore, we have $\Delta f(q_n)\rightarrow\Delta f(q)$, or $\Delta f$ is continuous. Now, since $r(\cdot)$ and $p(\cdot)$ are continuous functions we can conclude that
Also, $\Delta f(q_n) \geq 0$ as $f$ is an increasing function. Then, from (\ref{eq:Tnew}), we get that
\begin{align}
 T_\r f(q_n)& =C(q_n)+\b\E_Af(q_n+A) -  \int_{0}^{\beta \Delta f(q_n)} p_{\r}(u) du \nonumber\\
 & \rightarrow C(q)+\b\E_Af(q+A)-\int_{0}^{\beta \Delta f(q)} p_{\r}(u) du
 =T_\r f(q).\nonumber
\end{align}
Hence, $T_{\rho}f$ is a continuous function. This yields Statement~$1$ in the lemma.

Now, to prove the second part, assume that $\Delta f$ is an increasing function. First, we show that $\Delta T_\r f$ is an increasing function. Let $q > q'$. From (\ref{eq:Tnew}), for any $a < \bar{A}$ we can write
\begin{align*}
 &(T_\r f)(q+a)- (T_\r f) ((q-1)^++a)\nonumber\\
&\quad\quad - (T_\r f)(q'+a)+ (T_\r f)((q'-1)^++a)\nonumber\\
=&  C(q+a)-C((q-1)^++a) \nonumber\\
& - C(q'+a)+C((q'-1)^++a)\\
&  +\b\E_Af(q+a+A)-\b\E_Af((q-1)^++a+A)\\
&  - \b\E_Af(q'+a+A)+\b\E_Af((q'-1)^++a+A)\\
& -\int_{\b\Delta f(q'+a)}^{\b\Delta f(q+a)}p_{\r}(u)~du+\int_{\b \Delta f((q'-1)^++a)}^{\b \Delta f((q-1)^++a)}p_{\r}(u)~du\\
=& C(q+a)-C((q-1)^++a)\nonumber\\
 &- C(q'+a)+C((q'-1)^++a)\\
 & +\b\E_Af((q+a-1)^++A)-\b\E_Af((q-1)^++a+A)\\
 & - \b\E_Af((q'+a-1)^++A)+\b\E_Af((q'-1)^++a+A)\\
 & +\int_{\b\Delta f(q'+a)}^{\b\Delta f(q+a)}1-p_{\r}(u)~du+\int_{\b\Delta f((q'-1)^++a)}^{\b\Delta f((q-1)^++a)}p_{\r}(u)~du
\end{align*}
It can be easily verified that
%\begin{align*}
$\E_A(f(q+a-1)^++A)-\E_A(f(q-1)^++a+A)- \E_A(f(q'+a-1)^++A)+\E_A(f(q'-1)^++a+A) \geq 0$ as $f$ is increasing (due to Statement~$1$ of this lemma). From the assumption that $\Delta f$ is increasing, the last two terms in the above expression are also non-negative.
%\end{align*} Hence, we have $(T_\r f) (q) -
%which can be easily checked by considering the cases
%\begin{enumerate}
 %\item $1\leq q'<q$
% \item $q'\leq 1 \leq q$
% \item $q'<q<1$
%\end{enumerate}
Now, taking expectation on both sides, we obtain $\Delta T_{\r}f(q) - \Delta T_{\r} f(q')\geq \Delta C(q)-\Delta C(q')>0$. Therefore, from Statements $2$ and $3$ of Lemma~\ref{lem:optimalityMDP}, we have $$\hat\t_{\r}(q)-\hat\t_{\r}(q') = \Delta \hat{V}_{\r}(q) - \Delta \hat{V}_{\r}(q') \geq\Delta C(q)-\Delta C(q')>0.$$
Here, the last inequality holds since $C$ is a strictly convex increasing function.
\end{proof}
%We state a useful Corollary that defines the optimal policy of the agent.

%\begin{corollary}\label{cor:optbid}
%An optimal policy of the agent's decision problem \eqref{eq:optvalue} is given by
%$$\hat{\theta}_{\r}(q) = \beta \E_{A}\left[\hat{V}_\r(q+A) - \hat{V}_\r((q-1)^++A)\right]$$
%\end{corollary}
%The proof follows from Statement~$3$ of Lemma~\ref{lem:optimalityMDP} and Statement~$1$ of Lemma~\ref{lemma:rcontinuity}

%% file: 06-simulation.tex
\begin{figure*}[ht]
\centering
\begin{minipage}{.3\textwidth}
\centering
\includegraphics[width=1\columnwidth]{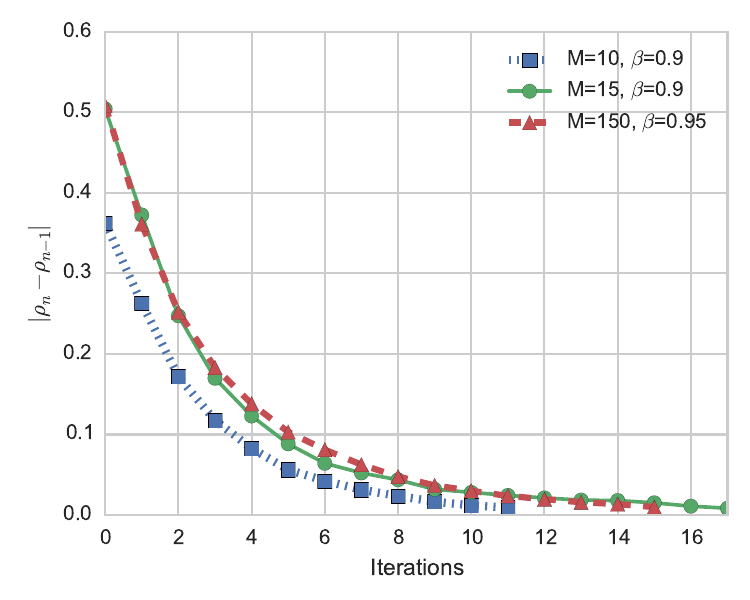}
\caption{Convergence to MFE bid distribution}
\label{fig:conv}
\end{minipage}\hfill
\begin{minipage}{.3\textwidth}
\centering
\includegraphics[width=1\columnwidth]{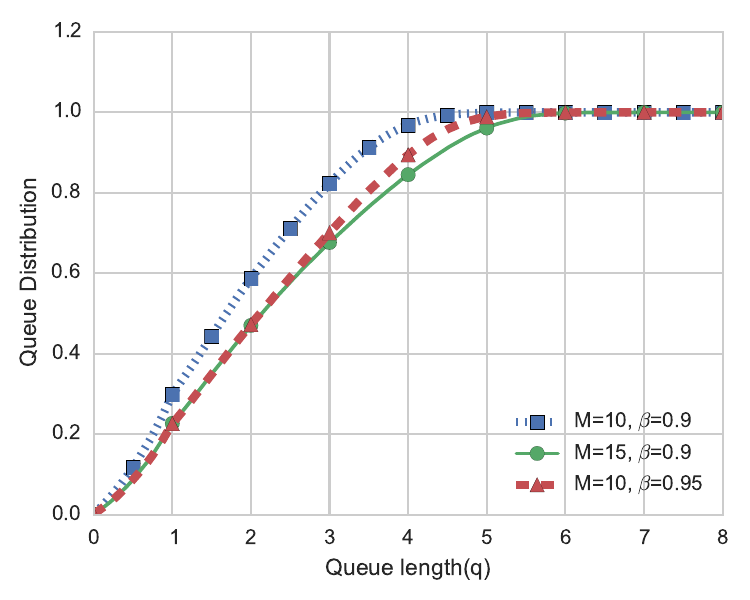}
\caption{MFE queue length distribution}
\label{fig:pi}
\end{minipage}\hfill
\begin{minipage}{.3\textwidth}
\centering
\includegraphics[width=\columnwidth]{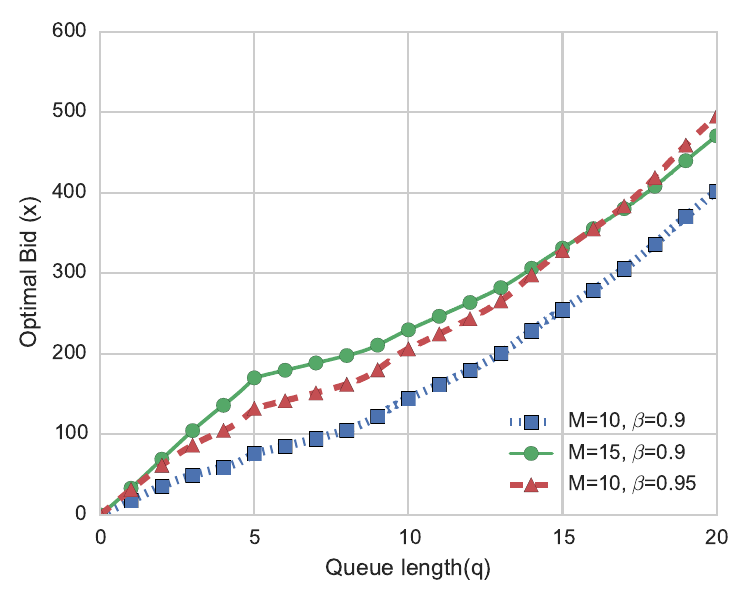}
\caption{MFE optimal bid function}
\label{fig:bids}
\end{minipage}
\end{figure*}
We now turn to computing the MFE distribution.  We simulate a large system with 100,000 users distributed among 10,000 cells with 10 users per cell. For simplicity of simulation, we truncate and discretize both state and bid spaces. The truncated state space is $ \mathcal{S} = \{0.01m, 0\leq m\leq 2000\}$, while the bid space is $\mathcal{X} = \{0.15m, 0\leq m\leq 3000\}$.  The job arrival and regeneration distributions are both chosen to be uniform over interval $[0,1]$.   The service rate of each base station is assumed to be $5$ units per time slot.  Finally, the holding cost function is chosen as $C(q) = q^2$.

Our simulation simply follows the choices made by each agent and calculating the empirical distribution that would result at each time step.  Let $\rho_0(x) = \min\{0.001 x,1\}, x \  in \mathcal{X}$ and $\Pi_0 = \Psi$. For every positive integer $n$, do the following:
\begin{enumerate}
\item Compute the optimal value function, $\hat{V}_n$, which is the unique solution to the following equation,
$$\hat{V}_n(q) = C(q) + \beta \E_A [ \hat{V}_n(q+A)] - \sum_{x \leq\b\Delta \hat{V}_n(q), x\in \mathcal{X}} p_{\rho_n}(x),$$
where $\Delta \hat{V}_n(q) = \E_A[ \hat{V}_n (q+A )] - \E_A[\hat{V}_n((q-1)^++A)]$. We apply value iteration (Section 6.10 \cite{puterman1994markov})
to compute an approximate solution to the above equation.
\item Compute the optimal bid function, $\hat{\theta}_n$ as $$\hat{\theta}_n(q) = \beta \E_A[ \hat{V}_n (q+A ) - \hat{V}_n((q-1)^++A)].$$
% \item Compute steady state distribution, $\Pi_n$ as follows:
% Let $p(q) = p_{\r_n}(\hat{\theta}_n(q)).$ Then,
% \begin{align*}
% &\Pi_n[q] = (1-\b)\Psi(q) + \b \sum_{q' \in \mathcal{S}} \Pi_{n-1}(q') \times \nonumber\\
%  &\left( p(q') \E_A[\1_{((q'-5)^++A = q)}] + (1- p(q')) \E_A[\1_{(q'+A =q)}] \right).
% \end{align*}
\item Next, all agents employ the optimal bid policy from Step~2 and update their states. Then we compute the candidate steady state distribution by evaluating emperical distribution of state.
\item Finally, compute the empirical bid distribution, $\rho_{n+1}(x) = \Pi_n[\hat{\theta}^{-1}_n([0,x])].$  If $||\rho_n - \rho_{n+1}|| < \epsilon$, then  $\rho = \rho_{n+1}$ and exit. Otherwise, set $n= n+1$ and go to Step~$1$.
\end{enumerate}
If the algorithm converges, then its output distribution $\rho,$ is an approximation of the MFE bid distribution.
%The convergence of the above algorithm is not well understood and may be subject of future study.

We simulated the algorithm for three set of parameters: $1.~(\beta = 0.9, M = 10)$,  $2.~(\beta = 0.95, M = 10)$ and $3.~(\beta = 0.9, M = 15)$.  Also, we chose the accuracy parameter $\epsilon =0.008$. Figure~\ref{fig:conv} shows that the algorithm converges in less than $20$ iterations in all three cases. In each iteration, Step~$1$ (value iteration) is the most computationally intensive.   It converges in $80$ recursions, with each recursion having to update $|\mathcal{S}|$ number of variables, and with each variable update requiring at most $|\mathcal{X}|$ number of arithmetic operations. All together, the computational complexity of each iteration is in the order of $80 \times |\mathcal{S}| \times |\mathcal{X}|$ arithmetic operations.

The queue length distributions at MFE are shown in Figure~\ref{fig:pi}.  We observe that the distribution curves exhibit a rightward shift with increase in $\beta$ or $M$.  Note that larger $\beta$ makes queues live longer without regeneration, while higher $M$ reduces each individual's average service rate. Hence, the queues get longer on average.  We show the optimal bid functions at MFE in Figure~\ref{fig:bids}. As expected from our analysis, the bid functions are monotonically increasing in queue length.

%% file: 07-conclusion.tex
Our algorithm for computing the MFE immediately suggests a simple implementation scheme.  Suppose each mobile device has a network interface manager on which the human user sets up priorities for different apps.  The interface manager also should be aware of the cost functions corresponding to the QoE of different apps.  Now, suppose that the base stations were to calculate the empirical bid distribution at each time instant, and return it to the interface manager.  The interface manager plays its best response to this bid distribution.  Value iteration could be done either on each device or at a data center and provided as a look up table to the interface manager.  The base stations combine all the bids to create a new empirical bid distribution.  Such a proceeding is essentially identical to the algorithm that we employed above, and would converge in a similar fashion.

We had assumed a single cost function for the agent (app), but as long as there are a finite set of cost functions (corresponding to a finite number of app types) we can incorporate the cost function as part of the state of the agent, with the cost function being chosen according to some distribution at each regeneration (corresponding to choosing to start a new app with some probability).  Such a modification causes no changes to any of our analysis.  Then the mean field bid distribution accounts for the distribution of cost functions (app popularities), and the agent takes a decision based on this distribution as before.  The only difference to the achieved equilibrium is that it now follows a weighted version of LQF, with the weights corresponding to the cost on QoE of the competing apps.

To summarize our work, we explored the question of whether it is possible to design an scheduling policy that allows for declaration of value by humans in the loop and effects on QoE, and which has the attractive properties of a (weighted) LQF service regime.   We used a mean field framework to show that as the number of agents in the system becomes large, this objective can indeed be fulfilled using a second price auction at each server.   Our design appears to lend itself well to implementation and this will be our future goal.

%Going further, our heuristic algorithm for computing the MFE employed a candidate MFE distribution estimated by a histogram of bids, which performed well in large scale simulations.  However, the question of convergence of such algorithm is an open question.  On the systems development side, the proposed system appears to lend itself well to implementation and has the potential to demonstrate how to integrate ideas of differential QoE and human input in a realistic setting.  This will be the focus of trials that we intend to conduct in our Smart Phone Laboratory in the future.

%A Mean Field Game Approach to QoE-Aware Cellular Scheduling

%% file: 08-appendix.tex
\appendices
\section{Proof of \Cref{lem:optimalityMDP}}\label{sec:optimalityMDP}
\ignore{
We first derive the formula in \cref{eq:opt-bid}.

\begin{lemma}\label{lem:optimalbid}
 Define $g(b,v)=r_\r(b)-p_\r(b)v$. Then ${v}\in\arg\min_{b\in\bb{R}^+}g(b,v)$.
 \end{lemma}
 \begin{proof}
  If $v\leq 0$, then $-p_\r(b)v\geq 0$ and $g(b,v)$ is increasing in $b$. Hence the minimum occurs at $b=0$. If $v>0$, then consider
  \begin{align*}
   g(b,v)-g(v,v)=&bp_\r(b)-\int_0^bp_\r(u)du-vp(b)+\int_0^vp_\r(u)du\\
   =&(b-v)p_\r(b)+\int_b^vp_\r(u)du\\
   =&\int_b^vp_\r(u)-p_\r(b)du\\
   \geq& 0
  \end{align*}
  with equality at $b=v$. Hence we have the desired result
 \end{proof}
}

%\begin{proof}[Details of Proof of \Cref{lem:optimalityMDP}]
We may rewrite the the definition of $T_{\r}$ in (\ref{eq:T}) as
\begin{align}\label{eq:Tn}
(T_{\r} f)(q) = \inf_{x \in \mathbb{R}^{+}} S_{f}(q,x)
\end{align}
where $S_{f}(q,x) = C(q) + r_{\r}(x) + \beta \E_{Q_1} [ f(Q_1)| q, x]$. Given the current state and action pair, $(q,x)$, the first two terms in  $S_{f}(q,x)$ constitute the current cost, while the last term is the future expected cost, where $Q_1$ is one-step future state variable. Further, from (\ref{eq:T}), we have
\begin{align}\label{eq:fcost}
  \E_{Q_1} [ f(Q_1)| q, x] & = (1- p_{\r}(x) \E_{A}[ f(q+A)] \nonumber\\
& + p_{\r}(x) \E_{A}[ f( (q-1)^{+} +A]. \nonumber
\end{align}
The proof proceeds through a verification of the assumptions of Theorem~$8.3.6$ in \cite{hernandez1999}. An exception is that action space in our case is not a compact set which violates Assumption~$8.3.1(a)$ in \cite{hernandez1999}. However, this assumption can be overridden if the statement of Lemma~$8.3.8 (a)$ in \cite{hernandez1999}, equivalently Condition $(3)$ below, holds true. Further, we desire to show the existence of a $j \in \mathbb{N}$ such that $T_{\r}^{j}$ is a contraction mapping.  Since Theorem~$8.3.6$ is derived for $j=1,$ we replace Assumption~$8.3.2(b)$ with Condition $(5)$ given below.

Now, we prove the following statements.
\begin{enumerate}
\item $C(q)+r_{\r}(x)$ is a continuous function in $x \in \mathbb{R}^{+}.$
\item  $\E_{Q_1} [ f(Q_1)| q, x]$ is continuous in $x \in \mathbb{R}^{+}$ for every $f \in \mathcal{V}$.
\item For any $f \in \mathcal{V}$, there exists a measurable function $\theta_{f}: \mathbb{R}^{+} \rightarrow \mathbb{R}^{+}$ such that $\theta_{f}(q)$ attains minimum in (\ref{eq:Tn}). Further, $S_{f}(q, \theta_{f}(q))$ is a measurable function for any $f \in \mathcal{V}.$
\item There exists a nonnegative constant $c_1$ such that $\sup_{x} |C(q)+r_{\r}(x)| \leq c_1 w(q)$ where $w(q) = \max\{C(q),1\}$.
\item There exists  $ j \in \mathbb{N}$ and  $c_2$ with $ c_2 < 1$ such that  $\beta^{j}\sup_{\vec{x}} \E_{Q_j}[ w(Q_j)|q,\vec{x}] \leq c_2 w(q) + c_3$ where $Q_j$ is $j$-step future state variable and $\vec{x}$ is a $j$-length sequence of actions.
\item The function $\E[ w(Q_1)|q,x]$ is continuous in $x \in \mathbb{R}^{+}$
\end{enumerate}

Conditions $(1)$ and $(2)$ are obvious from the continuity of $r_{\r}(x)$ and $p_{\r}(x)$. Further, as derived in (\ref{eq:opt-bid}), $\theta_{f}(q) = \Delta f(q)^+$ where $\Delta f(q) = \E_A (f(q+A)-f((q-1)^++A)).$ The measurability of functions $\theta_{f}(q)$ and $S_{f}(q, \theta(q))$ are evident from their definitions. Condition $(4)$ holds true from the definition of $w(q)$ and from the fact that $$r_\r(x)\leq \lim_{y\rightarrow\infty}r_\r(y) < (M-1) \int_0^\infty (1 -\r(x)) dx < (M-1)E$$ where the last inequality follows as $\r \in \mathcal{P}$. Condition~$(5)$ follows from the fact that,
\begin{align}
 \b^j\E_{Q_j}[w(Q_j)|q,\vec{x}]& \leq \b^j \max\{1, C(q+j\bar{A})\} \nonumber\\
 &\leq \b^j ( k_1 w(q) + k_2),\nonumber
\end{align}
where $\bar{A}$, as defined in Assumption~\ref{assum:arrival}, is the maximum arrival possible between any two adjacent auctions and $k_1>0, k_2$ are some constants independent of $j$. The above results follows from (\ref{eq:fcost}) and the definition of $w(q)$. Then, there exists a $j$ such that  $\b^jk_1 = c_2 <1$ and hence $(5)$ holds. Finally, the last condition follows from Condition $(2)$ as $w(q) \in \mathcal{V}$.

Given that the above conditions are met, we can prove the first statement of Lemma~\ref{lem:optimalityMDP}. The proof is essentially identical to that of Theorem~$8.3.6$ in \cite{hernandez1999}. The second statement of the lemma can be obtained by comparing (\ref{eq:T}) and (\ref{eq:bellman}). The last part of the lemma follows from (\ref{eq:opt-bid}).

\ignore{
To prove \eqref{eq:optcond0}, one may observe from \eqref{eq:Tnew} that
\begin{align}
C(q)  \leq (T_{\r}f)(q) \leq C(q)+ \beta \E_A f(q+A).
\end{align}
Here, the left most expression is positive. Also, the rightmost expression is bounded by some multiple of $w(q)$ since $A$ is a bounded random variable by Assumption~\ref{assum:arrival}. Together, we obtain  \eqref{eq:optcond0}. Here, the last inequality is due to $\r \in \mathcal{P}$.
 \Cref{eq:optcond2} holds true since

\begin{align}
\E_{Q_1}[f(q_1)|Q_0=q]&\leq\|f\|_w\E_{Q_1}[w(Q_1)|Q_0=q] \nonumber\\
&= \|f\|_w\left[p(b)\E_Aw((q-1)^++A)\right.\\
&\quad+\left.(1-p(b))\E_Aw(q+A)\right] \nonumber\\
&\leq \|f\|_w\left[\E_Aw(q+A)\right] \nonumber\\
&\leq \|f\|_wK_2w(q).\nonumber
\end{align}
}

%\end{proof}

\section{Proofs from \Cref{sec:mfe-proof}}\label{sec:prfconttheta}

In this section, we present details of proofs that were omitted from \Cref{sec:mfe-proof}.  We divide this section into parts based on the development of that section.

\subsection{Proofs Pertaining to  \Cref{sec:mfe-proof}-A: Step 1}\label{sec:step1}

\begin{lemma}\label{lem:theta-bound}
Suppose $\r_1,\r_2\in\mathcal{P}$. Then, $\|\hat\t_{\r_1}-\hat\t_{\r_2}\|_w\leq K\|\hat{V}_{\r_1}-\hat{V}_{\r_2}\|_w$
\end{lemma}
\begin{proof}
For any $q\in\bb{R}^+$, by the definition of $\hat \t_{\r}$ we have,
\begin{align*}
 &|\hat{\t}_{\r_1}(q)- \hat{\t}_{\r_2}(q)|\\
 =&|\b[\E_{A}[\hat{V}_{\r_1}(q+A)-\hat{V}_{\r_1}((q-1)^++A)\\
 &\quad-\hat{V}_{\r_2}(q+A)+\hat{V}_{\r_2}((q-1)^++A)]]|\nonumber\\
 \leq& \b\E_A|\hat{V}_{\r_1}(q+A)-\hat{V}_{\r_2}(q+A)|\\
 &\quad+\b\E_A|\hat{V}_{\r_1}((q-1)^++A)-\hat{V}_{\r_2}((q-1)^++A)|\nonumber\\
 \leq& \b\|\hat{V}_{\r_1}-\hat{V}_{\r_2}\|_w\E_A(w(q+A)+w((q-1)^++A))\\
 \leq&  K\|\hat{V}_{\r_1}-\hat{V}_{\r_2}\|_ww(q)
\end{align*}
\end{proof}

\begin{lemma}\label{lem:bellman-bound}
Let $\r\in\mathcal{P}$ and $f_1,f_2\in\mathcal{V}$. Then,
\begin{align}
 \|T_\r f_1-T_\r f_2\|_w\leq \hat K\|f_1-f_2\|_w %\Rightarrow \quad \mbox{(A)}
\end{align}
\end{lemma}
\begin{proof}
Using the characterization of $T_\r$ from \cref{eq:Tnew}, we have that, for any $q\in\bb{R}^+$
\begin{align*}
&|T_\r f_1(q)-T_\r f_2(q)|\\
%\leq& \beta|\E_A(f_1(q+A)-f_2(q+A))|+\left|\int_0^{\b\Delta f_1(q)}\rho^{M-1}(u)~du-\int_0^{\Delta f_2(q)}\r^{M-1}(u)~du\right|\\
\leq &\b \|f_1-f_2\|K_1w(q)+ \left|\int_{\b\Delta f_2(q)}^{\b \Delta f_1(q)}|\r^{M-1}(u)|~du\right|\\
\leq & \b \|f_1-f_2\|K_1w(q)+\beta |\Delta f_1(q)-\Delta f_2(q)|\\
\leq &\b (K_1+K_2)\|f_1-f_2\|w(q)
\end{align*}
\end{proof}

\subsection{Proofs Pertaining to  \Cref{sec:mfe-proof}-A: Step 2}\label{sec:step2}
\begin{proof}[Proof of \Cref{lem:queuelengthdist}]
For brevity, denote $\Pi_{\r,\t}(\cdot)$ be $\Pi(\cdot)$ and $\Upsilon_{\r,\t}^{(k)}=\Upsilon^{(k)}$ . Let $-\tau$ be the last time before $0$  the chain regenerated. We have
\begin{align}
 \Pi(B)=&\sum_{k=0}^\infty\pr(B,\tau=k)\\
 =&\sum_{k=0}^\infty\pr(\tau=k)\pr(B|\tau=k)
\end{align}
Since the regeneration events are independent of the queue-length and occur geometrically with probability $(1-\b)$, $\pr(\tau=k)=(1-\b)\b^k$. Hence,
\begin{align}
 \Pi(B)=&\sum_{k=0}^\infty(1-\b)\b^k\pr(Q_0\in B|\tau=k)\\
 =&\sum_{k=0}^\infty(1-\b)\b^k\E(\E(\mathbf{1}_{Q_0\in B}|\tau=k,Q_{-k}=Q)|\tau=k)\\
 =&\sum_{k=0}^\infty(1-\b)\b^k\E(\Upsilon^{(k)}(B|Q)|\tau=k)\\
 =&\sum_{k=0}^\infty(1-\b)\b^k\E_{\Psi_R}(\Upsilon^{(k)}(B|Q)).
\end{align}
since $Q_{-k}\sim\Psi_R$ given $\tau=k$.
\end{proof}

\begin{lemma}\label{lem:liminfupsilon}
$\liminf_{n\rightarrow\infty}\Upsilon_{\r_n}^{(k)}(B|q)\geq\Upsilon_{\r}^{(k)}(B|q)$
\end{lemma}
\begin{proof}
The proof proceeds through mathematical induction on $k$. For $k=0$, we have $\Upsilon_{\r_n}^{(0)}(B|q)=\mathbf{1}_{(q\in B)}$ and hence the hypothesis holds true. Suppose that the hypothesis is true till $k = m-1$. To prove the lemma, we just need to verify that the hypothesis holds for $k=m$. Let $\pr_{q, \r}(.)$ be the one step transition kernel of the queue dynamics conditioned on the following facts: the initial state is $q$, the bids are generated according to the optimal policy given by Corollary~\ref{cor:optbid} and no regeneration. Verify that $\pr_{q,\r_n}(\cdot)\implies \pr_{q,\r}(\cdot)$ by considering the integrals of a bounded continuous function.  Then, by Skorokhod representation theorem, there exists $X_n$ and $X$ on common probability space such that $X_n\sim \pr_{q,\r_n}$, $X\sim\pr_{q,\r}$ and $X_n\rightarrow X$ a.s. We have,
 \begin{align}
  \liminf\Upsilon_{\r_n}^{(m)}(B|q_n)=& \liminf \E(\Upsilon_{\r_n}^{(m-1)}(B|X_n))\\
  \geq& \E(\liminf \Upsilon_{\r_n}^{(m-1)}(B|X_n))\label{subeq:fatou}\\
  \geq & \E(\Upsilon_{\r}^{(m-1)}(B|X))\label{subeq:induc}\\
  =&\Upsilon_{\r}^{(m)}(B|q),
 \end{align}
 where \cref{subeq:fatou} follows from Fatou's lemma, and \cref{subeq:induc} follows from the induction hypothesis.
  \end{proof}

\subsection{Proofs Pertaining to  \Cref{sec:mfe-proof}-A: Step 3}\label{sec:step3}

\begin{proof}[Details of proof of \Cref{lem:PincludesF}]
To complete the proof, we need to show that the expected bid under the cumulative distribution function $\hat{\r}$ is bounded from above by a constant that is independent of $\hat{\r}$. To that end, define a new Markov random process $\tilde{Q}_{k}$ with the probability transition matrix
\begin{align}\label{eq:probmeas}
 \pr(\tilde{Q}_{k+1}\in B|\tilde{Q}_k=q)=\b\mathbf{1}_{(q+\bar{A}\in B)}+(1-\b)\Psi_R(B)
\end{align}
where $\bar{A}$ is the maximum possible arrival between any two consecutive auction instants. The process $\tilde{Q}_k$ has an invariant distribution which is given by,
\begin{align}
 \tilde{\Pi}(B)=\sum_{k=0}^\infty(1-\b)\b^k\E_{\Psi_R}(\mathbf{1}_{(q+k\hat{A})\in B}).
\end{align}
The proof of the above result is identical to that of Lemma~\ref{lem:queuelengthdist}.
For any $q$ given, the above probability measure~(\ref{eq:probmeas}) stochastically bounds the probability measure in \cref{eq:trprobs}. Therefore, it can be shown that $\tilde{\Pi}$ stochastically dominates $\Pi_{\r}$ for all $\r \in \mathcal{P}$, i.e, $\Pi_\r\preccurlyeq \tilde{\Pi}$.

Now, the expected value of the optimal bid function $\hat{\t}_\r(q)$ under $\Pi_\r$ satisfies,
\begin{align}
 \E_{\Pi_\r}[\hat{\t}_\r(q)] \leq &\E_{\tilde{\Pi}}[\hat{\t}_\r(q)]\\
 \leq& \E_{\tilde{\Pi}} [\hat{V}_\r(q+\bar{A})]\\
 \leq&\sum_{k=0}^\infty(1-\b)\b^k\E_{\Psi_R}(\hat{V}_\r(q+(k+1)\bar{A}))
\end{align}
Above, the first inequality follows from stochastic dominance of $\tilde{\Pi}$ and the second inequality is due to the definition of optimal bid function.

From (\ref{eq:bellman}), we can observe that for any $\r$, $\hat{V}_\r(q)\leq \sum_{k=0}^\infty\b^k C(q+k\bar{A})$ independent of $\r$. Since $C(q) \in O(q^m)$ for some $m$, we have $\hat{V}_\r(q) \in O(q^m)$. Then, $\E_{\Psi_R}(\hat{V}_\r(q+(k+1)\hat{A})) \in O(k^m)$ as the moments of $\Psi_R$ are bounded. This directly gives that  $\E_{\Pi_\r}[\hat{\t}_\r(q)]$ is bounded by the some constant that is independent of $\r$ and, hence independent of $\hat{\r}$.
\end{proof}

\begin{lemma}\label{lem:point-uniform}
In $\mathcal{P}$, pointwise convergence implies uniform convergence.
\end{lemma}
\begin{proof}
Let $\r_n, \r \in \mathcal{P}$ and $\r_n \rightarrow \r$ point-wise. Given $\epsilon>0$, choose $L$ large enough so that $\r(L)>1-\epsilon$. Since $\r$ is continuous function by definition, it is uniformly continuous on the compact set $[0,L]$. Therefore, we can construct a sequence $0=x_1<x_2<\cdots<x_k=L$ such that and $|\r(x_{i+1})-\r(x_i)|<\epsilon$. Let $J$ be large enough so that for all $n>J$, $|\r(x_i)-\r_n(x_i)|<\epsilon$ for all $i$. For any $y$ such that $x_i<y<x_{i+1}$,
  \begin{align}
    &|\r(y)-\r_n(y)|\\
    <& |\r(y)-\r(x_i)|+|\r(x_i)-\r_n(x_i)|\\
    &\quad+|\r_n(y)-\r_n(x_i)|\nonumber\\
    <&|\r(x_{i+1})-\r(x_i)|+|\r(x_i)-\r_n(x_i)|\\
    &\quad+ |\r_n(x_{i+1})-\r_n(x_i)|\nonumber\\
    <&2|\r(x_{i+1})-\r(x_i)|+|\r(x_i)-\r_n(x_i)|+2\epsilon\\
    <&5\epsilon
  \end{align}
  While if $L<y$, then
  \begin{align}
   &|\r(y)-\r_n(y)|\\
   &< |\r(y)-\r_n(L)|+|\r_n(L)-\r(L)|+|\r(y)-\r(L)|\\
  & <1-\r(L)+\epsilon+\epsilon +1-\r(L)\\
   &<4\epsilon.
  \end{align}
Therefore, $|\r(y) -\r_n(y)| < 5\epsilon$ for all $n > J$ and hence $\r_n$ converges to $\r$ uniformly.
\end{proof}

\subsection{Proofs Pertaining to  \Cref{sec:mfe-proof}-B}\label{sec:partB}

\begin{proof}[Proof of \Cref{lem:bound-on-Pi}]
We know that $\Pi([a,b]|\r,\t)=\sum_{k\geq 0}(1-\b)\b^k\E_{\Psi_R}(\Upsilon^{(k)}_{\r}([a.b]|Q_0))$. %Let $\{Q_t\}$ be the queue-length stochastic process with $Q_0\sim {\Psi_R}$.
%We first condition on the event that there are no regenerations until time $k$.
Let $A_k$ be the net arrivals and $D_k$ be the net departures till time $k$. Then,
  \begin{align}
\Upsilon^{(k)}_{\r}([a,b]|Q_0) &= \E(\mathbf{1}_{(Q_0+A_k-D_k\in[a,b])}|Q_0)\\
&=\E( \E(\mathbf{1}_{(Q_0+A_k-D_k\in[a,b])}|D_k,Q_0)|Q_0)\\
&=\E( \E(\mathbf{1}_{(A_k\in[a-Q_0+D_k,b-Q_0+D_k])}|Q_0,D_k)|Q_0)\\
&\leq c_{1} \cdot(b-a).
  \end{align}
The above results hold since the random variable $A_k$ is independent of $Q_0$ and $D_k$ for any $k$ and it has a bounded density function. Therefore, $\E_{\Psi_R}(\Upsilon^{(k)}_{\r}([a.b]|Q_0))\leq c_\cdot(b-a)$ for all $k>0$. For $k=0$, we know that $\Psi_R$ has a bounded density which implies $\Psi_R([a,b])\leq c_1{\psi}\cdot(b-a)$. These two results prove that there is a large enough $c$ such that $\Pi_\r([a,b])<c\cdot(b-a)$.
\end{proof}

%% file: arxiv-draft.bbl
\begin{thebibliography}{10}
\providecommand{\url}[1]{#1}
\csname url@samestyle\endcsname
\providecommand{\newblock}{\relax}
\providecommand{\bibinfo}[2]{#2}
\providecommand{\BIBentrySTDinterwordspacing}{\spaceskip=0pt\relax}
\providecommand{\BIBentryALTinterwordstretchfactor}{4}
\providecommand{\BIBentryALTinterwordspacing}{\spaceskip=\fontdimen2\font plus
\BIBentryALTinterwordstretchfactor\fontdimen3\font minus
  \fontdimen4\font\relax}
\providecommand{\BIBforeignlanguage}[2]{{%
\expandafter\ifx\csname l@#1\endcsname\relax
\typeout{** WARNING: IEEEtran.bst: No hyphenation pattern has been}%
\typeout{** loaded for the language `#1'. Using the pattern for}%
\typeout{** the default language instead.}%
\else
\language=\csname l@#1\endcsname
\fi
#2}}
\providecommand{\BIBdecl}{\relax}
\BIBdecl

\bibitem{Tass}
L.~Tassiulas and A.~Ephremides, ``Stability properties of constrained queueing
  systems and scheduling policies for maximum throughput in multihop radio
  networks,'' in \emph{IEEE Transactions on Automatic Control}, vol.~37, 1992,
  pp. 1936--1948.

\bibitem{LinShr04}
X.~Lin and N.~B. Shroff, ``Joint rate control and scheduling in multihop
  wireless networks,'' in \emph{in Proceedings of IEEE Conference on Decision
  and Control}, 2004, pp. 1484--1489.

\bibitem{ErySri06}
A.~Eryilmaz and R.~Srikant, ``Joint congestion control, routing and mac for
  stability and fairness in wireless networks,'' \emph{IEEE Journal on Selected
  Areas in Communications}, vol.~24, pp. 1514--1524, 2006.

\bibitem{NeeMod05}
M.~J. Neely, E.~Modiano, and C.~ping Li, ``Fairness and optimal stochastic
  control for heterogeneous networks,'' in \emph{Proceedings of IEEE INFOCOM},
  2005, pp. 1723--1734.

\bibitem{HaSen12}
S.~Ha, S.~Sen, C.~Joe-Wong, Y.~Im, and M.~Chiang, ``Tube: time-dependent
  pricing for mobile data,'' in \emph{SIGCOMM}, 2012, pp. 247--258.

\bibitem{krishna2009auction}
V.~Krishna, \emph{{Auction theory}}.\hskip 1em plus 0.5em minus 0.4em\relax
  Academic press, 2009.

\bibitem{LasLio07}
J.-M. Lasry and P.-L. Lions, ``Mean field games,'' \emph{{J}apan {J}ournal of{
  M}athematics}, 2007.

\bibitem{TemBou09}
H.~Tembine, J.-Y. Le~Boudec, R.~El-Azouzi, and E.~Altman, ``Mean field
  asymptotics of {M}arkov decision evolutionary games and teams,'' in
  \emph{GameNets}, 2009, pp. 140--150.

\bibitem{adlakha2011equilibria}
S.~Adlakha, R.~Johari, and G.~Y. Weintraub, ``Equilibria of dynamic games with
  many players: Existence, approximation, and market structure,'' \emph{Journal
  of Economic Theory}, vol. 156, pp. 269--316, 2015.

\bibitem{BorSun11}
V.~Borkar and R.~Sundaresan, ``Asymptotics of the invariant measure in mean
  field models with jumps,'' in \emph{49th Annual Allerton Conference on
  Communication, Control, and Computing}, 2011, pp. 1258--1263.

\bibitem{XuHaj12}
J.~Xu and B.~Hajek, ``The supermarket game,'' in \emph{ISIT}, 2012.

\bibitem{IyeJoh14}
K.~Iyer, R.~Johari, and M.~Sundararajan, ``Mean field equilibria of dynamic
  auctions with learning,'' \emph{Management Science}, vol.~60, no.~12, pp.
  2949--2970, 2014.

\bibitem{LiBha15}
J.~Li, R.~Bhattacharyya, S.~Paul, S.~Shakkottai, and V.~Subramanian,
  ``Incentivizing sharing in realtime {D2D} streaming networks: A mean field
  game perspective,'' in \emph{Proceedings of IEEE Infocom}, Hong Kong, 2015.

\bibitem{ManRam14}
M.~Manjrekar, V.~Ramaswamy, and S.~Shakkottai, ``A mean field game approach to
  scheduling in cellular systems,'' in \emph{Proceedings of INFOCOM}, 2014, pp.
  1554--1562.

\bibitem{Gross02}
M.~Grossglauser and D.~N.~C. Tse, ``Mobility increases the capacity of ad hoc
  wireless networks,'' \emph{IEEE/ACM Trans. Netw.}, vol.~10, no.~4, Aug. 2002.

\bibitem{strauch1966negative}
R.~E. Strauch, ``{Negative dynamic programming},'' \emph{The Annals of
  Mathematical Statistics}, vol.~37, no.~4, pp. 871--890, 1966.

\bibitem{meyn2009markov}
S.~P. Meyn, R.~L. Tweedie, and P.~W. Glynn, \emph{{Markov chains and stochastic
  stability}}.\hskip 1em plus 0.5em minus 0.4em\relax Cambridge University
  Press, 2009, vol.~2.

\bibitem{BenBou08}
M.~Benaim and J.-Y.~L. Boudec, ``A class of mean field interaction models for
  computer and communication systems,'' \emph{Performance Evaluation}, vol.~65,
  no.~11, pp. 823--838, 2008.

\bibitem{hernandez1999}
O.~Hernandez-Lerma and J.-B. Lasserre, \emph{Further topics on discrete-time
  Markov control processes}, ser. Applications of mathematics.\hskip 1em plus
  0.5em minus 0.4em\relax Springer, 1999.

\bibitem{billingsley2009convergence}
P.~Billingsley, \emph{Convergence of probability measures}.\hskip 1em plus
  0.5em minus 0.4em\relax Wiley-Interscience, 2009, vol. 493.

\bibitem{graham1994chaos}
C.~Graham and S.~M{\'e}l{\'e}ard, ``Chaos hypothesis for a system interacting
  through shared resources,'' \emph{Probability Theory and Related Fields},
  vol. 100, no.~2, pp. 157--174, 1994.

\bibitem{puterman1994markov}
M.~L. Puterman, \emph{{Markov decision processes: Discrete stochastic dynamic
  programming}}.\hskip 1em plus 0.5em minus 0.4em\relax John Wiley \& Sons,
  Inc., 1994.

\end{thebibliography}
